\documentclass[xcolor=svgnames, 11pt, reqno]{amsart}

\usepackage[T1]{fontenc}
\usepackage{amsmath}
\usepackage{amssymb}
\usepackage{amsthm}
\usepackage[left=2.75cm, right=2.75cm, top=3cm, bottom=2.5cm]{geometry}
\usepackage[bottom,flushmargin,hang,multiple]{footmisc}
\allowdisplaybreaks
\usepackage{chngcntr}
\usepackage{apptools}

\AtAppendix{\counterwithin{lemma}{section}}
\AtAppendix{\counterwithin{theorem}{section}}
\AtAppendix{\counterwithin{proposition}{section}}
\AtAppendix{\counterwithin{corollary}{section}}
\AtAppendix{\counterwithin{remark}{section}}

\usepackage{fancyhdr}
\usepackage{enumitem}
\usepackage{comment}
\usepackage{nicefrac}
\usepackage{bm}
\usepackage{mathrsfs}
\usepackage{graphicx}
\usepackage{subcaption}
\usepackage[utf8]{inputenc}
\usepackage{cancel}
\usepackage{mathtools}
\usepackage{tikz}
\usetikzlibrary{cd,decorations.pathreplacing,matrix,arrows,positioning,automata,shapes,shadows,calc,fadings,decorations,snakes,through,intersections}
\usepackage{float}
\usepackage{epsfig}
\usepackage{tabularx}
\usepackage{booktabs}
\usepackage{pifont}

\AtBeginDocument{%
   \def\MR#1{}
}

\theoremstyle{definition}

\newtheorem{example}{Example}

\newtheorem{definition}{Definition}
\newtheorem*{axiom}{Axiom}
\newtheorem*{property}{Property}

\usepackage{extarrows}
\theoremstyle{plain}
\newtheorem{theorem}{Theorem}
\newtheorem{lemma}{Lemma}
\newtheorem{proposition}{Proposition}
\newtheorem{corollary}{Corollary}
\theoremstyle{remark}
\newtheorem{remark}{Remark}

\theoremstyle{definition}

\newcommand{\E}{\mathbb{E}}
\newcommand{\R}{\mathbb{R}}
\newcommand{\A}{\mathbb{A}}
\newcommand{\F}{\mathcal{F}}
\newcommand{\X}{\mathcal{X}}
\newcommand{\G}{\mathcal{G}}

\newcommand{\Xvec}{\boldsymbol{X}}

\newcommand{\Yvec}{\boldsymbol{Y}}

\newcommand{\Zvec}{\boldsymbol{Z}}

\newcommand{\evec}{\boldsymbol{e}}

\usepackage{todonotes}

\usepackage{setspace}
\usepackage{xr-hyper}
\usepackage{hyperref}
\hypersetup{
    pdftitle={Prova},
    pdfauthor={},
    pdfmenubar=false,
    pdffitwindow=true,
    pdfstartview=FitH,
    colorlinks=true,
    linkcolor=blue,
    citecolor=green,
    urlcolor=cyan
}

\definecolor{darkblue}{rgb}{0.1,0.1,0.9}
\definecolor{darkred}{rgb}{0.9,0.1,0.1}

\hypersetup{colorlinks, allcolors= blue}

\providecommand{\MR}[1]{}

\providecommand{\MR}{\relax\ifhmode\unskip\space\fi MR }

\begin{document}

\title[Allocation Mechanisms with Frictions]{Allocation Mechanisms in Decentralized\vspace{0.2cm}\\Exchange Markets with Frictions\vspace{0.6cm}}


\author[Mario Ghossoub, Giulio Principi, and Ruodu Wang]{Mario Ghossoub\vspace{0.1cm}\\University of Waterloo\vspace{0.7cm}\\Giulio Principi\vspace{0.1cm}\\New York University\vspace{0.7cm}\\Ruodu Wang\vspace{0.1cm}\\University of Waterloo\vspace{1.2cm}\\This draft: \today\vspace{0.2cm}}

\address{{\bf Mario Ghossoub}: University of Waterloo -- Department of Statistics and Actuarial Science -- 200 University Ave.\ W.\ -- Waterloo, ON, N2L 3G1 -- Canada}
\email{\href{mailto:mario.ghossoub@uwaterloo.ca}{mario.ghossoub@uwaterloo.ca}\vspace{0.2cm}}

\address{{\bf Giulio Principi}: New York University -- Department of Economics -- 19 West 4th Street -- New York, NY 10012 -- USA}
\email{\href{mailto:gp2187@nyu.edu}{gp2187@nyu.edu}\vspace{0.2cm}}

\address{{\bf Ruodu Wang}: University of Waterloo -- Department of Statistics and Actuarial Science -- 200 University Ave.\ W.\ -- Waterloo, ON, N2L 3G1 -- Canada}
\email{\href{mailto:wang@uwaterloo.ca}{wang@uwaterloo.ca}}


\thanks{\textit{Key Words and Phrases: Allocation Mechanisms, Risk-Sharing Rules, Market Frictions, Frictional Participation, Decentralized Risk Sharing, Peer-to-Peer Insurance.\vspace{0.15cm}}}

\thanks{\textit{JEL Classification:} C02, D49, D51, D86, D89, G22.\vspace{0.15cm} }

\thanks{\textit{2020 Mathematics Subject Classification:} 46A20, 46A22, 46N10, 47H99, 47N10, 91B05, 91B30, 91G99.\vspace{0.15cm}}


\thanks{We are grateful to Luciano De Castro, Fabio Maccheroni, and Marco Scarsini for comments and suggestions. We thank audiences at the 2023 workshop on the Foundations and Applications of Decentralized Risk Sharing (FADeRiS 2023) at KU Leuven. Mario Ghossoub acknowledges financial support from the Natural Sciences and Engineering Research Council of Canada (NSERC Grant No.\ 2018-03961 and 2024-03744). Ruodu Wang acknowledges financial support from the Canada Research Chairs (CRC-2022-00141) and the Natural Sciences and Engineering Research Council of Canada (NSERC Grant No.\  2024-03728). Giulio Principi is grateful for the financial support provided by the Henry M.\ MacCracken Fellowship at New York University.\vspace{0.2cm}}


\begin{abstract}  
The classical theory of efficient allocations of an aggregate endowment in a pure-exchange economy has hitherto primarily focused on the Pareto-efficiency of allocations, under the implicit assumption that transfers between agents are frictionless, and hence costless to the economy. In this paper, we argue that certain transfers cause frictions that result in costs to the economy. We show that these frictional costs are tantamount to a form of subadditivity of the cost of transferring endowments between agents. We suggest an axiomatic study of allocation mechanisms, that is, the mechanisms that transform feasible allocations into other feasible allocations, in the presence of such transfer costs. Among other results, we provide an axiomatic characterization of those allocation mechanisms that admit representations as robust (worst-case) linear allocation mechanisms, as well as those mechanisms that admit representations as worst-case conditional expectations. We call the latter \textit{Robust Conditional Mean Allocation} mechanisms, and we relate our results to the literature on (decentralized) risk sharing within a pool of agents.  
\end{abstract}

\maketitle
\thispagestyle{empty}

\newpage
\section{Introduction}

Two fundamental concerns of a decentralized pure-exchange economy, be it price-mediated or of the barter type, are the willingness of agents to engage in trade and the efficiency of allocations of the aggregate social endowment. This basic economic problem of distribution is at the core of welfare economics, and an implicit underlying assumption made in the literature on efficient allocations of an aggregate endowment is that transfers between agents are frictionless, and hence costless to the economy. It is our primary contention that certain types of transfers between agents are indeed costly to the economy. Specifically, we posit that in many instantiations of the notion of a pure-exchange economy, certain kinds of subsidization, namely initial transfers of state-contingent endowments to any agent with zero initial endowment, ultimately result in costs to the economy. We understand these as generic frictions in the market that translate into trading or reallocation costs, which we call \textit{frictional costs}. We show that the existence of these frictional costs in the economy is equivalent to a form of subadditivity of the cost of transferring endowments between agents in the economy. An immediate example is any platform that centralizes exchanges between participants, in return for a fixed participation fee, whence the subadditivity of participation costs ensues naturally. 

\medskip

This then begs the question of what are the effects of such frictional costs on the shape and structure of allocation mechanisms; that is, the mechanisms that redistribute feasible allocations of the aggregate endowment into other feasible allocations. Clearly, an immediate implication of frictional costs in this economy is that trade will result in a cost to be deducted from the initial aggregate endowment. This does not mean that the allocation is inefficient or harmful for the participating agents; as we will see from examples in the paper, agents will still benefit from participation in case an appropriate cost is deducted. 
We wish to go a step further and provide a general axiomatic characterization of allocation mechanisms in the presence of such frictional costs. Hitherto, there has been nearly no systematic axiomatic examination of allocation mechanisms. Two notable exceptions are the recent work of Jiao et al.\ \cite{Jiaoetal2023} and Dhaene et al.\ \cite{Dhaeneetal2023} in the context of risk sharing rules that we discuss below.

\medskip

Indeed, the interest of the literature on Pareto efficiency in pure-exchange economies under uncertainty is in the existence (and often the comonotonicity) of Pareto-efficient allocations. The literature is vast, and its groundwork was laid down by Borch \cite{Borch1962a} and Wilson \cite{Wilson1968}, who showed that when agents have risk-averse expected-utility preferences, each individual's consumption is a nondecreasing deterministic function of the aggregate endowment at a Pareto optimum. Chateauneuf et al.\ \cite{CCK97}, Landsberger and Meilijson \cite{LandsbergerMeilijson1994c}, and Carlier et al.\ \cite{CarlierDanaGalichon2012} extend this fundamental insight to economies where agents have probabilistically sophisticated preferences that are consistent with second-order stochastic dominance. Characterizations of efficiency beyond expected-utility theory (EU) were also considered. For instance, economies with Maxmin Expected Utility (MEU) mutiple-prior preferences \`a la Gilboa and Schmeidler \cite{GilboaSchmeidler1989maxmin} were examined by Dana \cite{Dana2002, Dana2004} and De Castro and Chateauneuf \cite{decastro2011}. 
Risk exchange with quantile agents are studied by Embrechts et al.~\cite{embrechts2018quantile, embrechts2020quantile}, notably without any notion of convexity. 
Economies with non-probabilistic uncertainty of the Choquet-Expected Utility (CEU) type (as in Schmeidler \cite{schmeidler89}) were studied by Chateauneuf et al.\ \cite{Chateauneufetal2000}, Dana \cite{Dana2004}, De Castro and Chateauneuf \cite{decastro2011}, and Beissner and Werner \cite{BeissnerWerner2023}. Dana and LeVan \cite{DanaLeVan2010} and Ravanelli and Svindland \cite{RavanelliSvindland2014} looked at the case of economies with variational preferences (see Maccheroni, Marinacci, and Rustichini \cite{Maccheronietal2006}). Jouini et al.\ \cite{JouiniSchachermayerTouzi2008}, Filipovi\'c and Svindland \cite{filipovic2008optimal}, and Ghossoub and Zhu \cite{ghossoub2026efficiency} studied the case of economies with law-invariant monetary utilities (see Delbaen \cite{DelbaenOsaka}). Economies with general ambiguity-sensitive convex preferences were examined by Strzalecki and Werner \cite{StrzaleckiWerner2011}. Pareto-efficiency of risk exchange in risk-sharing markets with risk minimization criteria rather than welfare maximization criteria has been studied in parallel to the above literature. We refer to R\"uschendorf \cite{Ruschendorf2013} for an overview.

\medskip

The literature has also been interested in efficient allocations when the aggregate endowment is constant. In this context, comonotonicity is equivalent to full insurance, and trading is seen as betting rather than hedging. In the case of a constant aggregate endowment, the Pareto optimality of full-insurance allocations was studied by Malinvaud \cite{Malinvaud1972,Malinvaud1973} under EU; Billot et al.\ \cite{Billotetal2000} and Kajii and Ui \cite{KajiiUi2006} under MEU; Chateauneuf et al.\ \cite{Chateauneufetal2000}, Billot et al.\ \cite{Billotetal2002}, and Dominiak et al.\ \cite{Dominiak2012} under CEU; Beissner et al.\ \cite{BBG24} under Rank-Dependent Utility; Rigotti et al.\ \cite{Rigottietal2008} for a general class of ambiguity-sensitive convex preferences; and Ghirardato and Siniscalchi \cite{GhirardatoSiniscalchi2018} for a general class of non-convex preferences. 

\medskip

Two points that have hitherto not been the primary focus of the literature are (i) the axiomatic study of allocation mechanisms themselves -- that is, those mapping that transform feasible allocations of the aggregate endowment into other feasible allocations -- and (ii) the presence of transfer costs in the economy. Two notable exceptions to (i) are the recent work of Jiao et al.\ \cite{Jiaoetal2023} and Dhaene et al.\ \cite{Dhaeneetal2023} in the context of risk sharing rules. In the former study, the authors propose an axiomatization of the so-called \textit{Conditional Mean-Risk Sharing} (CMRS) rule and \textit{Generalized Conditional Mean-Risk Sharing} (Generalized CMRS) rule. In the latter, the authors provide an axiomatization of the so-called \textit{Quantile-Based Risk-Sharing} (QBRS) rule. The CMRS was popularized in the theory of risk exchange by the work of Denuit and Dhaene \cite{DenuitDhaene2012}
and Denuit et al.\ \cite{Denuitetal2022}, building upon the work of Landsberger and Meilijson \cite{LandsbergerMeilijson1994c} and Dana and Meilijson \cite{DanaMeilijson2003}. It is a risk reallocation mechanism that assigns to each individual the conditional expectation of their initial risky endowment, conditioning on the aggregate risk in the economy. Hence, the only information required for the risk allocation is that generated by the aggregate risk. Knowledge of the realization of each individual risk is not needed. The Generalized CMRS was introduced by Jiao et al.\ \cite{Jiaoetal2023} as an extension that accounts for additional information availability, extraneous to the information generated by the aggregate risk.

\medskip

The CMRS, Generalized CMRS, and QBRS rules are linear allocation mechanisms, and hence do not reflect frictional allocation costs in the economy. Indeed, they are designed to satisfy the so-called \textit{Full Allocation} condition (sometimes called the \textit{Self-Financing} condition), which axiomatically imposes on an allocation mechanism to redistribute the entirety of the initial aggregate endowment. This is incompatible with the existence of frictional transfer costs, which, as argued above, result in a reallocation cost to be deducted from the initial aggregate endowment. Additionally, and by design, the CMRS, Generalized CMRS, and QBRS rules do not lead to the kind subsidization discussed above that we argue leads to frictional costs, namely, the initial transfers of state-contingent endowments to any agent with zero initial endowment. Indeed, these risk-sharing rules assign, by construction, a zero endowment to any individual with an initial zero endowment.

\medskip

Extrapolating from the natural property of cost subadditivity, we propose to capture the aforementioned frictional costs through an axiom that we call \textit{Frictional Participation}, which we show is equivalent to subadditivity of a certain cost function that we introduce. In addition to standard monotonicity, continuity, and anonymity axioms, we suggest two axioms introduced by Jiao et al.\ \cite{Jiaoetal2023} in the context of risk sharing to deal with the effect of differing information sets on risk allocations: \textit
{Information Anonymity} and \textit
{Information Backtracking}. We show how combinations of these axioms yield a crisp representation of allocation mechanisms as worst-case, robust linear allocation mechanisms (Theorem \ref{RGRS_prop}), as well as worst-case conditional expectations (Theorem \ref{thm1}). The latter representation is of particular importance in a context of risk-sharing within a pool of agents, as it provides an extension of the CMRS rule of Denuit and Dhaene \cite{DenuitDhaene2012} and the Generalized CMRS rule of Jiao et al.\ \cite{Jiaoetal2023}. We call such mechanisms \textit{Robust Conditional Mean Allocation} mechanisms, and in the context of risk-sharing rules we refer to them as \textit{Robust Conditional Mean Risk Sharing} rules. As a corollary, we provide an axiomatization of (linear) Conditional Mean Allocation mechanisms (Proposition \ref{subj_CMRS}), as well as a subjective version of the CMRS rule.

\medskip

As an illustration, we discuss two interesting and practically relevant examples of robust allocation mechanisms: the Mean-Deviation allocation mechanism and the Expected-Shortfall allocation mechanism, based on  the theory of  risk/deviation measures; see Rockafellar et al.~\cite{RockafellarUryasev06} and McNeil et al.~\cite{mcneil2015quantitative}. In each example, we propose a model of market allocation based on the given robust allocation mechanism, and we show how the global frictional cost in the market can be parameterized by a single parameter. This parameter can be understood as related to the fees imposed by the allocator on the participants.

\medskip

On a technical level, this paper contributes to the mathematical literature on convex duality. In particular, we provide envelope representation results for superlinear operators mapping from Dedekind complete Riesz spaces to Dedekind complete Riesz spaces. These results are then applied in the context of finite Cartesian products of $L^1$-spaces, and conditional-expectation representations of operators. From this perspective, the closest work to ours are the studies on conditional risk measures. We discuss the relevant mathematical literature and our contribution to it in Section \ref{MathLit}.

\medskip

The rest of this paper is set out as follows. Section \ref{Economy} introduces our model, as well as necessary definitions. Section \ref{SecRep} contains our main results on the axiomatic characterizations of robust allocation mechanisms. In Section \ref{SecRiskSharing}, we discuss the use of our results in the theory of decentralized risk sharing within a pool of agents (i.e.,m peer-to-peer insurance). In Section \ref{SecApplication}, we examine two practical examples of robust allocation mechanisms. Section \ref{MathLit} expands on the mathematical literature on convex duality to which we contribute. Section \ref{SecConc} concludes. The \hyperlink{LinkToAppendix}{Appendices} contain some background material, including our aforementioned mathematical contribution, as well as the proofs of our main results.

\bigskip
\section{The Economy}
\label{Economy}

\subsection{Allocations}
Let $(\Omega,\mathcal{F},\mathbb{P})$ be a probability space, and denote by $\Sigma$ the set of all sub-$\sigma$-algebras of $\mathcal{F}$. For a given $\G \in \Sigma$, let $L^1(\G)$ be the space of all $\G$-measurable and $\mathbb{P}$-integrable random variables. To ease notation we set $\X=L^1(\F)$ and let $\mathcal{X}_+$ denote the cone of nonnegative elements of $\X$. We use the symbol $\mathbb{E}[\cdot]$ to denote expectation with respect to $\mathbb{P}$. We treat $\mathbb{P}$-a.s.\ equal random variables as identical and denote by $\sup X$ the essential supremum of any $X\in \X$. Vectors of random variables will be denoted in bold by $\Xvec$. We equate constants with constant random variables. Whenever we consider expectations taken with respect to a probability measure $Q$ rather than $\mathbb{P}$, such expectations will be denoted by $\mathbb{E}^Q\left[\cdot\right]$.

\medskip

Consider a pure exchange economy under uncertainty, where $\Omega$ denotes the set of future states of the world. There are $n\in \mathbb{N}$ agents, each with an initial state-contingent endowment $X_i \in \X$, for $i \in [n] := \{1, \ldots, n\}$. We assume $n\ge 3$ throughout, which is realistic for applications in decentralized risk-exchange economies.  For any random variable $S \in \X$, the set of \textit{allocations} of $S$ is denoted by
\begin{equation}\label{allocations}
\mathbb{A}_n^{=}(S):=\left\lbrace (Y_1,\ldots,Y_n) \in \X^n:\sum_{i=1}^nY_i= S\right\rbrace,
\end{equation}

\medskip

\noindent and the set of \textit{sub-allocations of} $S$ is denoted by
\begin{equation}
\label{suballocations}
\mathbb{A}^{\leqslant}_n(S):=\left\lbrace (Y_1,\ldots,Y_n) \in \X^n:\sum_{i=1}^nY_i\leq S\right\rbrace.
\end{equation}

\medskip

\noindent For each $(Y_1,\ldots,Y_n) \in \mathbb{A}^{\leqslant}_n(S)$, the quantity $S - \sum_{i=1}^nY_i$ can be interpreted as a frictional cost, corresponding to the unallocated part of $S$. To ease the notation we set $S^{\Xvec}:=\sum_{i=1}^n X_i$, for all $\Xvec\in \X^n$.

\medskip
\subsection{Allocation Mechanisms}
In the context of decentralized risk sharing within a pool of agents, Denuit et al.\ \cite{Denuitetal2022} introduced the notion of a risk-sharing rule, as a mapping that transforms a feasible allocation of an aggregate risk into another feasible allocation. In a more general context of trade in a pure-exchange economy, we call this a simple allocation mechanism.

\medskip

\begin{definition}
 A \textit{simple allocation mechanism} is a mapping $H: \X^n \to  \X^n$ such that for all $\Xvec:=(X_1,\ldots,X_n)\in \X^n$,
 $$H(\Xvec)=\left(H_1(\Xvec), \ldots, H_n(\Xvec)\right)\in \A^=_n(S^{\Xvec}).$$
\end{definition}

\medskip

Jiao et al.\ \cite{Jiaoetal2023} introduced the notion of a generalized risk-sharing rule, as an extension of the notion of a risk-sharing rule that accounts for the available information, which might be more than the observable value of the state-contingent aggregate risk. In the context of the allocation mechanisms we study in this paper, such information determine what risk is not gonna
be shared among the participants. This information is interpreted as \textit{target information} and it is modeled by a sigma-algebra $\G\in \Sigma$, called an information set. In a more general context of trade in a pure-exchange economy, we call this a costless allocation mechanism. 

\medskip

\begin{definition}
 A \textit{costless allocation mechanism} is a mapping $H:\X^n\times \Sigma\to  \X^n$ such that for all $\Xvec:=(X_1,\ldots,X_n)\in \X^n$ and all $\G\in \Sigma$,
 $$H(\Xvec|\G)=\left(H_1(\Xvec|\G), \ldots, H_n(\Xvec|\G)\right)\in \A^{=}_n(S^{\Xvec}).$$
\end{definition}

\medskip

Note that we are using $H(\mathbf X|\mathcal G)$ instead of   $H(\mathbf X,\mathcal G)$, because a primary example of $ H(\mathbf X|\mathcal G)$  is the usual conditional expectation $\E[\mathbf X|\mathcal G]$, which will be clear in the later sections. 

\medskip

The definition of a generalized risk-sharing rule and that of a costless allocation mechanism require frictionless transfers in the economy, so that any reallocation of the initial endowments preserves the aggregate endowment, for any given information set. In order to account for potential costly frictions, we introduce the notion of an \textit{allocation mechanism} below.

\medskip

\begin{definition}
An \textit{allocation mechanism} is a mapping $H: \X^n\times \Sigma \to  \X^n$ such that for all $\Xvec:=(X_1,\ldots,X_n)\in \X^n$ and all $\G\in \Sigma$,
$$H(\Xvec|\G)=(H_1(\Xvec|\G), \ldots, H_n(\Xvec|\G))\in \A^{\leqslant}_n(S^{\Xvec}).$$ 
 
\noindent The set of allocation mechanisms will be denoted by $\mathcal{AM}$.
\end{definition}

\medskip

For each information set $\G$ and each vector of initial endowments $\Xvec:=(X_1,\ldots,X_n)\in \X^n$, any $H \in \mathcal{AM}$ induces a total cost to the economy, given by 
$$S^{\Xvec} - \sum_{i=1}^n H_i(\Xvec|\G) \geq 0.$$

\noindent This motivates the following notions. 

\begin{definition}
A \textit{local frictional cost} is a mapping
\begin{align*}
\mathcal{C}:  
\mathcal{AM} \times  \X^n \times \Sigma \times [n]^2 & \to \X\\
(H,\Xvec,\G, \left(i,j\right)) &\mapsto \mathcal{C}^{\Xvec,\G}_{i,j}(H) := X_i+X_j - H_i(\Xvec|\G) - H_j(\Xvec|\G).
\end{align*}
A \textit{global frictional cost} is a mapping
\begin{align*}
\mathcal{C}:  
\mathcal{AM} \times  \X^n \times \Sigma & \to \X_+\\
(H,\Xvec,\G) &\mapsto \mathcal{C}^{\Xvec,\G}(H) := S^{\Xvec} - \sum_{i=1}^n H_i(\Xvec|\G).
\end{align*}
\end{definition}

\noindent For each $(\Xvec,\G) \in \X^n \times \Sigma$, the quantity $\mathcal{C}^{\Xvec,\G}(H)$ measures the frictional cost associated with the allocation mechanism $H \in \mathcal{AM}$. This cost corresponds to the part of $S^{\Xvec}$ that cannot be allocated through $H$ and is ultimately borne by the economy. Clearly, $H(\Xvec|\mathcal{G})\in \mathbb{A}^{=}_n(S^{\Xvec})$ if and only if $\mathcal{C}^{\Xvec|\G}(H) = 0$.

\begin{table}[H]
\centering
\label{tab:notation}
\renewcommand{\arraystretch}{1.14}
\setlength{\tabcolsep}{6pt}
\begin{tabularx}{\textwidth}{l X}
\hline
\textbf{Mathematical Notation} & \textbf{Meaning}\\
\hline\smallskip
$(\Omega,\mathcal{F},\mathbb{P})$ & Base probability space\smallskip\\
$\mathcal{X}=L^1(\mathcal{F})$ & $\F$-measurable \& $\mathbb{P}$-integrable random variables\smallskip\\
$\mathcal{X}_+$ & Cone of nonnegative elements of $\X$\smallskip\\
$\Xvec=(X_1,\dots,X_n)\in \mathcal{X}^n$, $S^{\Xvec}=\sum_{i=1}^n X_i$ & Initial endowments and their aggregate\smallskip\\
$\mathcal{A}^{=}_n(S)$, $\mathcal{A}^{\le}_n(S)$ & (Sub-)allocations: $\sum_i Y_i=S$ (resp.\ $\le S$)\smallskip\\
$\Sigma$ & Set of all sub-\,$\sigma$-algebras of $\mathcal{F}$\smallskip\\\
$\sigma\left\lbrace S \right\rbrace$ & $\sigma$-algebra generated by a random variable $S$\smallskip\\
$\mathcal{G}\in\Sigma$, $\G^{\Xvec} := \sigma\left\{\G, \sigma\left\{S^{\Xvec}\right\}\right\}$ & Information set; augmented info with $S^{\Xvec}$\smallskip\\
$H(\Xvec\mid \mathcal{G})=(H_1(\Xvec\mid \mathcal{G}),\dots,H_n(\Xvec\mid \mathcal{G}))\in \mathcal{X}^n$ & Allocation mechanism\\
$\mathcal{AM}$ & Set of allocation mechanisms\smallskip\\
$C^{\Xvec,\mathcal{G}}(H)=S^{\Xvec}-\sum_i H_i(\Xvec\mid \mathcal{G})$ & Global frictional cost\smallskip\\
$C_{i,j}^{\Xvec,\mathcal{G}}(H)=X_i+X_j-H_i(\Xvec\mid \mathcal{G})-H_j(\Xvec\mid \mathcal{G})$ & Local frictional cost for pair $(i,j)$\smallskip\\
$\bigtriangleup(\F)$ & Probability measures on $\mathcal{F}$\smallskip\\
$\bigtriangleup^{\mathbb{P}}(\F|\G)$ & The set of $Q\in \bigtriangleup(\F)$ s.t. $Q\ll \mathbb{P}$ and $Q|_{\G}=\mathbb{P}|_{\G}$\smallskip\\
$\mathbb{E}^{\mathbf{Q}}[\Xvec\mid \mathcal{G}]$ & Cond. expectations vector $(\mathbb{E}^{Q_i}[X_i\mid \mathcal{G}])_{i=1}^n$\smallskip\\
\hline
\end{tabularx}
\caption{Notation summary}
\end{table}

\bigskip

\section{Allocation Mechanisms with Frictions}
\label{SecRep}
In what follows, for all $\Xvec\in \X^n$, $S\in \X$, and $\G\in \Sigma$, we let 

\vspace{-0.1cm}

$$\G^{\Xvec} := \sigma\left\{\G, \sigma\left\{S^{\Xvec}\right\}\right\}
\ \textnormal{and} \ 
\G^S:=\sigma\left\lbrace \G, \sigma\left\{S\right\}\right\rbrace.$$

\medskip

\noindent A bijective function $\pi:\left\lbrace 1,\ldots,n\right\rbrace\to \left\lbrace 1,\ldots,n\right\rbrace$ will be called a \textit{permutation}. We endow $\X$ and $L^1(\G)$ with the $\mathbb{P}$-almost sure order, for all $\G\in \Sigma$. For $i \in [n]$, let $\evec^{(i)}$ denote the unit vector along the $i^{th}$ axis, that is

\vspace{-0.1cm}

$$\evec^{(i)} := (0, \ldots, 0, 1, 0, \ldots, 0),$$ 

\medskip

\noindent where the $i^{th}$ component is $1$. We denote by $\bigtriangleup(\F)$ the set of all probability measures over $(\Omega,\mathcal{F})$. The space of linear operators from $\X$ to $\X$ will be denoted by $\mathcal{L}$. For all sequences $(X_m)_{m\in \mathbb{N}}$ in $\X$ and all $X\in \X$, the notation $X_m\downarrow X$ means that $(X_m)_{m\in \mathbb{N}}$ is decreasing and $\inf_{m\in \mathbb{N}}X_m=X$. Analogously $X_m\uparrow X$ reads as $(X_m)_{m\in \mathbb{N}}$ is an increasing sequence such that $\sup_{m\in \mathbb{N}}X_m=X$.

\medskip
\subsection{Axioms and Properties}
In this section we provide axioms that will be employed to characterize specific classes of allocation mechanisms. We make a distinction between the terms \textit{axioms} and \textit{properties}. Axioms indicate conditions that we argue are desirable for a function to be an allocation mechanism in the economy described above. Properties refer to more technical requirements. We fix an arbitrary mapping $H: \X^n\times \Sigma \to  \X^n$. All conditions expressed below refer to this mapping.

\medskip

\begin{axiom}[{\bf Internal Fairness -- IF}]
For all $\G \in \Sigma$, $\Xvec \in \X^n$, and $i,j\in \left[n\right]$,
$$X_i\geq X_j \, \implies \, H_i\left(\Xvec|\G\right) \geq H_j\left(\Xvec|\G\right).$$
\end{axiom}

\medskip

\noindent As the name implies, IF is an intra-allocation fairness condition. In relative terms, the larger the initial endowment, the more the agent receives. The next two axioms express anonymity principles. The first, \textit{Agent Anonymity}, is a standard condition that expresses the indifference to re-labelings, or in more colloquial terms, that the names of the agents do not matter for the allocation. It is referred to as Symmetry by Jiao et al.\ \cite{Jiaoetal2023}, and Reshuffling by Denuit et al.\ \cite{Denuitetal2022}.

\medskip

\begin{axiom}[{\bf Agent Anonymity -- AA}]
For all permutations $\pi: [n] \to [n]$, $\G \in \Sigma$, $\left(X_1,\ldots,X_n\right)\in \X^n$, and $i\in \left[n\right]$, 
$$H_i\left(X_{\pi(1)},\ldots,X_{\pi(n)}|\G\right)=H_{\pi(i)}\left(X_1,\ldots,X_n|\G\right).$$
\end{axiom}

\medskip

\begin{axiom}[{\bf Operational Anonymity -- OA}]
For all $\G \in \Sigma$, $\Xvec \in \X^n$, and $i,j \in [n]$, if $\Yvec = \Xvec -X_j \, \evec^{(j)} + X_j \evec^{(i)}$, then
$$H_k\left(\Yvec|\G\right) = H_k\left(\Xvec|\G\right), \ \forall \, k \in [n]\setminus\{i,j\}.$$
\end{axiom}

\medskip

\noindent OA, introduced by Jiao et al.\ \cite{Jiaoetal2023}, dictates that the allocation to one agent is unaffected by transfers between two other agents. 
Here, the term \textit{anonymity} refers to the fact that two agents do not need to disclose whether they are merged into one agent or not, as this merger does not affect the allocation to other agents.
For instance, these two agents may be two accounts of the same person in an online platform.
From OA, it follows that for all $i\in \left[n\right]$, all $\left(X_1,\ldots,X_n\right)\in \X^n$, and all $\G \in \Sigma$,  
$$H_i\left(X_1,\ldots,X_n|\G\right)=H_i\left(0,\ldots,0,X_i,S^{\Xvec}-X_i,0,\ldots,0|\G\right).$$

\medskip

The following axiom captures the frictions in this market. It stipulates that transfers of state-contingent endowments to any agent with no initial endowment (i.e., state-contingent subsidization) are costly to the economy. 

\medskip

\begin{axiom}[{\bf Frictional Participation -- FP}]
For all $\G \in \Sigma$, $Z \in \X$, $i \in [n]$, and $\Xvec \in \X^n$ such that $X_j=0$ for some $j \in [n]$, 
$$\Yvec = \Xvec + Z \evec^{(j)} - Z \evec^{(i)} \, \implies \, \mathcal{C}_{i,j}^{\Yvec,\G}(H)\geq \mathcal{C}_{i,j}^{\Xvec,\G}(H).$$
\end{axiom}

\medskip

\noindent FP is pivotal for our analysis, as it is one of the main conditions that will lead to a representation of $H$ as a so-called \textit{robust allocation mechanism} (Definition \ref{RBGRS} and Theorem \ref{RGRS_prop}), and it distinguishes the framework of this paper from the related literature.

\medskip

\begin{axiom}[{\bf Scale Invariance -- SI}]
For all $\G \in \Sigma$, $i \in [n]$, $\alpha \in [0,1]$, and $\Xvec \in \X^n$ such that $X_j=0$ for some $j \in [n]$,
$$\Yvec = \Xvec + \alpha X_i \evec^{(j)} - \alpha X_i \evec^{(i)} \, \implies \, H_j(\Yvec|\G) =
\alpha  H_i(\Xvec|\G).$$
\end{axiom}

\medskip

\noindent SI limits the extent of the impacts of frictions on the economy, in particular imposing the absence of proportional frictions. As such, it could be interpreted as a fairness condition. Indeed, suppose that agent $i$, with initial endowment $X_i$, decides to split her initial endowment proportionally with agent $j$, who has a zero initial endowment, whereby the latter receives the proportion $\alpha  X_i$, and $(1-\alpha)  X_i$ is retained by the former. SI then stipulates that if the allocation to agent $j$ (that is, $H_j(\Yvec|\G)$) differs from a proportion $\alpha$ of the starting allocation of agent $i$ (that is, $\alpha H_i(\Xvec|\G)$), then this could be considered unfair either to agent $i$ or to agent $j$. 

\medskip

Note that the SI is not implied by the previous axioms. For instance one could consider the allocation mechanism $H:\X^n\to \X^n$ defined as $H\left(\Xvec\right):=S^{\Xvec}/n$ for all $\Xvec\in \X^n$. Such an allocation mechanism violates SI while satisfying IF, AA, OA, and FP. Examples can also be provided for IF. One can consider the allocation mechanism $H:\X^n\times \Sigma\to \X^n$ defined as  $H(\Xvec|\G):=\mathbb{E}\left[-|\Xvec| \, \vert \, \G^{\Xvec}\right]$, for all $\Xvec\in \X^n$ and $\G\in \Sigma$. It is straightforward to see that such an allocation mechanism satisfies AA, OA, FP, and SI but violates IF.

\medskip

\begin{remark}
To highlight the interplay between some of the introduced axioms, note that if $H \in \mathcal{AM}$ satisfies OA and FP, then for all $\Yvec\in \X^n$ and $i,j\in \left[n\right]$,
$$\mathcal{C}^{\Yvec|\G}(H)\geq \mathcal{C}^{\Xvec|\G}(H),$$

\noindent where $\Xvec:=\Yvec-Y_j\evec^{(j)}+Y_j\evec^{(i)}$. To see this formally, fix $\Yvec\in \X^n$ and $i,j\in \left[n\right]$, and let $\Xvec:=\Yvec-Y_j\evec^{(j)}+Y_j\evec^{(i)}$. Then $S^{\Yvec} = S^{\Xvec}$ and $\G^{\Yvec} = \G^{\Xvec}$. Moreover, OA and Lemma \ref{cost_sub_is_split_robu} in the Appendix imply that for each $\mathcal{G}\in \Sigma$, 
$$\sum_{k=1}^n H_k\left(\Xvec|\mathcal{G}\right)\geq \sum_{k=1}^n  H_k\left(\Yvec|\mathcal{G}\right).$$
Consequently, for each $\mathcal{G}\in \Sigma$, 
$$S^{\Xvec} - \sum_{k=1}^n H_k(\Xvec|\mathcal{G})
=S^{\Yvec} - \sum_{k=1}^n H_k(\Xvec|\mathcal{G})
\leq S^{\Yvec} - \sum_{k=1}^n H_k\left(\Yvec|\mathcal{G}\right).$$
Hence,
$\mathcal{C}^{\Yvec|\G}(H) = S^{\Yvec} - \sum_{k=1}^n H_k\left(\Yvec|\G\right) 
\geq  
S^{\Xvec} - \sum_{k=1}^n H_k\left(\Xvec|\G\right)
=
\mathcal{C}^{\Xvec|\G}(H).$
\end{remark}

\bigskip
\smallskip

\noindent{\bf{Terminology.}}
We will refer to the combination of axioms IF, AA, OA, FP, and SI as \textit{Allocation Axioms}.

\bigskip

Next, we report a few technical conditions that we shall refer to as \textit{properties}.

\medskip

\begin{property}[{\bf Zero Preserving -- ZP}]
For all $\G \in \Sigma$ and $\Xvec \in \X^n$, if $X_i=0$ for some $i \in [n]$ then $H_i\left(\Xvec|\mathcal{G}\right)=0$.
\end{property}

\medskip

\noindent ZP imposes on an allocation mechanism to provide a zero allocation to any agent who did not contribute to the aggregate endowment. 

\medskip

Below, we present two properties first introduced by \cite{Jiaoetal2023} in the context of risk sharing rules. Such properties describe the role of information as modeled by $\G$ and $\G^{\Xvec}$ for allocation mechanisms.

\medskip

\begin{property}[{\bf Information Anonymity -- IA}]
For all $\G \in \Sigma$, $\Xvec \in \X^n$, and $i \in [n]$, 
$$\sigma\left\{H_i\left(\Xvec|\G\right)\right\} \subseteq \G^{\Xvec}.$$
\end{property}

\medskip

\begin{property}[{\bf Information Backtracking -- IB}]
For all $\G \in \Sigma$, $i\in \left[n\right]$, and $\Xvec \in \X^n$
$$\sigma\{X_i\} \subseteq \G^{\Xvec} \implies H_i\left(\Xvec|\G\right) = X_i.$$
\end{property}
\medskip

\noindent Property IA requires that the realized allocation be exclusively determined by the aggregate endowment and the target information $\G$, and not by the 
specific initial individual endowments. The measurability with respect to $\G$ reflects that in some applications, more information than simply the realized value of the total endowment may be available, and the mechanism designer may wish to allocate endowments relying on such information. On the other hand, Property IB states that when the initial endowments are sufficiently informative, no re-allocation of the initial endowments is needed.

\medskip

\begin{property}[{\bf Continuity from Above -- CA}]
If $\left(X_m\right)_{m\in \mathbb{N}}\in \X^{\mathbb{N}}$ is such that $X_m\downarrow X$ and $|X_m|\leq Y$ for some $X,Y\in \mathcal{X}$ and all $m\in \mathbb{N}$, then

\vspace{-0.2cm}

$$H_i\left(X_m\, \evec^{(i)} + \left(S-X_m\right)\evec^{(j)}|\G\right)\downarrow H_i\left(X\, \evec^{(i)} + \left(S-X\right) \evec^{(j)}|\G\right),$$

\medskip

\noindent for all $\left(S,\G\right) \in \mathcal{X} \times \Sigma$, and all $i,j\in [n]$ with $i\neq j$.
\end{property}

\medskip

\noindent{\bf{Terminology.}}
We will refer to the combination of the properties ZP, IA, IB, and CA as \textit{Allocation Properties}.

\medskip

\subsection{An illustrative example}

We provide an example illustrating how frictional costs can arise naturally as platform fees, in peer-to-peer pure-exchange markets mediated by an exchange platform, thereby motivating the role of frictions for allocation mechanisms. 

\medskip

Consider a pure-exchange economy with $n$ expected-utility (EU) agents, with utilities $(u_1,\ldots,u_n)$. Given initial endowments $\Xvec=(X_1,\ldots,X_n)\in \X^n$, an allocation $\Yvec=(Y_1,\ldots,Y_n)\in \mathcal{A}^{\le}_n(S)$ is said to be:

\medskip

\begin{itemize}
\item \textit{Individually rational} (IR) if it incentivizes the agents to participate in the market, that
is, for all $i \in [n]$, $\mathbb{E}[u_i(Y_i)]\geq \mathbb{E}[u_i(X_i)]$.

\bigskip

\item \textit{Pareto optimal} (PO) if it is IR and there does not exist any IR allocation $\Zvec\in \X^n$ such that, for all $i\in [n]$,
\[
\mathbb{E}[u_i(Z_i)]\geq \mathbb{E}[u_i(Y_i)],
\]

\smallskip

\noindent with at least one strict inequality.
\end{itemize}

\medskip

\begin{example}[Frictions]Suppose that the individual endowments $X_1,\dots,X_n$ of agents $1,\dots,n$ are independent and identically distributed (iid), positive, and nonconstant.   These conditions are not essential for the conclusion, but help to simplify calculations. 
These $n \geq 3$ individuals wish to enter into an exchange market mediated by a dedicated platform (the market maker). We will refer to the platform as agent $n+1$, which we assume to have a $0$ initial endowment. The aggregate endowment in this economy is denoted by $S := \sum_{i=1}^n X_i = \sum_{i=1}^{n+1} X_i$.

\medskip

Suppose that all agents $i\in [n+1]$ are EU agents with  the same  strictly concave utility function $u$ in the CRRA family, indexed by $\gamma>0$. That is, $u(x)= {x^{1-\gamma}}/{(1-\gamma)}$, for $x>0$ if  $\gamma\ne1$; and $u(x)=\log(x)$, for $x>0$ if $\gamma=1$.  Assume  that $\E[|u(S)|]<\infty$. 

\medskip

We claim that there exists $\epsilon>0$ such that 
\begin{equation}\label{eq:R1-example}
\Big(\left(1-\epsilon\right) S/n,\dots,\left(1-\epsilon\right) S/n, \epsilon S\Big)
\end{equation}
is a PO allocation for the $n+1$ agents. To show this, we first verify individual rationality. Clearly, the utility of agent $n+1$ improves from $u(0)$ to $\E[u(\epsilon S)]$. For the other agents, note that for any convex function $\varphi$,
$$
\varphi\left(S/n\right)
=\varphi\left(\frac{1}{n}\sum_{i=1}^n X_i\right)
\leq
\frac{1}{n}\sum_{i=1}^n \varphi(X_i).
$$

\smallskip

\noindent Hence, taking expectations (assuming integrability) and using the fact that the initial endowments are iid, it follows that, for each $i \in [n]$, 
$$
\E[\varphi\left(S/n\right)]
\leq 
\frac{1}{n} \, \sum_{i=1}^n \E[\varphi(X_i)]
=\E[\varphi(X_1)].
$$

\noindent Moreover, since $n \geq 3$ and since the random variables are assumed nonconstant (more precisely, there is no $c \in \mathbb{R}$ or $i \in [n]$ such that $\mathbb{P}(X_i=c) = 1$), then the inequality is strict for some convex function $\varphi$. Therefore, $S/n$ is strictly smaller than  $X_1$ in the convex order. Consequently, the strict concavity of $u$ implies that $\E[u(S/n)] >\E[u(X_1)]$.
Since $u$ is continuous, the Monotone Convergence Theorem implies that $\E[u((1-\epsilon)S/n)] \to \E[u(S/n)]$ as $\epsilon\downarrow 0$.
Therefore, for $\epsilon>0$ small enough, \eqref{eq:R1-example} yields an IR allocation.

\medskip

We now verify that the allocation  \eqref{eq:R1-example} is indeed PO.  To show this, note that it suffices to show that \eqref{eq:R1-example} maximizes the quantity
\begin{equation}
\label{eq:R1-example-2}
\sum_{i=1}^{n} \E[u(Y_i)] + \beta \, \E[u(Y_{n+1})],
\end{equation}

\noindent over $(Y_1,\dots,Y_{n+1})\in \mathbb A^=_{n+1}(S)$, where $\beta =\left(\frac{1-\epsilon}{n\epsilon}\right)^{-\gamma}>0$. The concavity of $u$ implies that it suffices to consider $Y_1=\dots=Y_n$ when maximizing \eqref{eq:R1-example-2}. 
Hence, it suffices to show that  for each $s>0$, the function
$$
y\mapsto nu(y) + \beta u(s-ny)
$$
is maximized at $y=(1-\epsilon)s/n$. 
This is a direct consequence of the first-order condition. Therefore, the allocation in \eqref{eq:R1-example} is PO.

\medskip

Next, we analyze the value of $\epsilon>0$. The individual rationality constraint is given by 
$$
\E[u((1-\epsilon) \, S/n)] \ge \E[u(X_i)],
$$
for each agent $i\in [n]$. 
If the platform (agent $n+1$) is a monopoly, it may be able to choose $\epsilon=\epsilon_0$ such that 
$$
\E[u((1-\epsilon_0) \, S/n)] = \E[u(X_i)].
$$

\smallskip

\noindent In practice, $\epsilon$ can be chosen slightly smaller than $\epsilon_0$, so that agent $i$ is strictly better off by joining the exchange platform. For simplicity, let us assume that $\epsilon_0$ is chosen by the above equality, meaning that all agents $i \in [n]$ are made indifferent between entering the peer-to-peer arrangement and not doing so, so that all potential consumer (utility) surplus is absorbed by the platform.

\medskip

Consider the situation where two agents $i$ and $j$ merge into one. This resulting agent is the new agent $i$ with endowment $X_i+X_j$. The new agent $j$ has zero endowment. In the above simple model, it is natural (although not strictly necessarily) to allocate $2 \, (1-\epsilon) \, S/n$ to the new   agent $i$ and $0$ to the new agent $j$, with $\epsilon>0$ to be recalculated.
In this case,  an acceptable $\epsilon>0$ needs to satisfy
$$
\E[u (2\,(1-\epsilon)\,S/n)] \ge  \E[u(X_i+X_j)],
$$

\medskip

\noindent to ensure individual rationality for the new agent $i$. Using the properties of the CRRA utility function, and  the fact that $X_i+X_j$ is strictly smaller than $2X_1$ in the convex order, we obtain 
$$
\E[u(2 (1-\epsilon)\, S/n)] 
\ge \E[u(X_i+X_j)]
> 
\E[u(2\,(1-\epsilon_0) \, S/n)]. 
$$

\smallskip

\noindent This immediately gives 
$\epsilon < \epsilon_0$ 
as a necessary condition for the new agent $i$ to remain in the risk-sharing pool, i.e., to satisfy individual rationality. 
As a consequence, the platform will need to allocate more payoffs in total to the $n$ agents and take less for itself.
In particular, the frictional cost to agents $i$ and $j$ decreases from $X_i+X_j-2(1-\epsilon_0)S$
to $X_i+X_j-2(1-\epsilon)S$.

\medskip

The above analysis reflects the idea of Frictional Participation (FP): by combining endowments, the frictional cost reduces. Therefore, the above simple example shows that FP naturally appears in a Pareto-optimal and individually rational allocation.
\end{example}

\medskip

\subsection{Frictional Participation and Convexity of Costs}

As mentioned above, our axiom FP is meant to capture frictions in the economy that result from our basic assumption that transfers of state-contingent endowments to any agent with no initial endowment (i.e., state-contingent subsidization) are costly to the economy. These frictions result in a form of convexity of the cost functional associated with a given allocation mechanism. We formalize this below.

\medskip

\begin{proposition}\label{PropConvexCost}
If $H \in \mathcal{AM}$ satisfies FP, OA, SI, and AA, then for all $i,j \in [n]$ and all $\G \in \Sigma$, the mappings 
$\Xvec \mapsto \mathcal{C}^{\Xvec,\G}_{i,j}(H)$
and $\Xvec \mapsto\mathcal{C}^{\Xvec,\G}(H)$
satisfy
$$\mathcal{C}^{\lambda\Xvec+(1-\lambda)\Yvec,\G}_{i,j}(H)\leq \lambda \mathcal{C}^{\Xvec,\G}_{i,j}(H)+(1-\lambda)\mathcal{C}^{\Yvec,\G}_{i,j}(H)$$
and
$$\mathcal{C}^{\lambda\Xvec+(1-\lambda)\Yvec,\G}(H)\leq \lambda \mathcal{C}^{\Xvec,\G}(H)+(1-\lambda)\mathcal{C}^{\Yvec,\G}(H),$$

\medskip

\noindent for all $\lambda\in [0,1]$, and all $\Xvec,\Yvec\in \X^n$ with $S^{\Xvec}=S^{\Yvec}$.
\end{proposition}

\medskip

Therefore, if an allocation mechanism $H \in \mathcal{AM}$ satisfies FP, SI, OA, and AA, then the local and global frictional costs are convex over the collection of all feasible allocations that have the same total. They are in fact also positively homogeneous (by inspection of the proof), and hence subadditive.

\medskip
\subsection{Representation of Allocation Mechanisms}
In this section we introduce robust allocation mechanisms and their representation through the previous axioms and properties. We start with a more abstract setting before moving on to a refinement of such robust mechanisms, as expressed by nonlinear conditional expectations. 

\medskip
\subsubsection{Abstract Robust Mechanisms}

Here we provide an abstract representation result for scale-invariant allocation mechanisms exhibiting FP. We start with the following definition.

\medskip

\begin{definition}\label{RBGRS}
An allocation mechanism $H\in \mathcal{AM}$ is said to be a \textit{robust allocation mechanism} if
    \[
H(\Xvec|\G)=\begin{bmatrix} \inf\limits_{L_1\in D_1(\Xvec,\G)}L_1(X_1)\\
    \vdots\\
\inf\limits_{L_n\in D_n(\Xvec,\G)}L_n(X_n)
    \end{bmatrix},
    \]

    \medskip
    
    \noindent for all $\Xvec\in \X^n$ and $\G\in \Sigma$, and where each $D_i(\Xvec,\G)$ is a convex set of linear mappings $L:\X\to \X$ such that $L(X)\geq H_i(X\evec^{(i)}+\left(S^{\Xvec}-X\right)\evec^{(j)}|\G)$ for all $j\in \left[n\right]$ with $i\neq j$ and all $X\in \X$.
\end{definition}

\medskip

\begin{remark}[Robust Allocation Mechanisms and Convex Costs]
It is immediate to see that for any robust allocation mechanism, both the local and global frictional costs are convex (in the sense of Proposition \ref{PropConvexCost}). 
\end{remark}

\medskip

It is important to observe that for a robust allocation mechanism $H$, the sets $D_i(\Xvec,\G)$ can be substituted by the following supporting sets:

\vspace{-0.1cm}

$$D_i\left(H_i,S^{\Xvec},\G\right):=\bigg\{ L\in \mathcal{L}:\forall X\in \X, \forall j\in \left[n\right]\setminus \left\lbrace i\right\rbrace,\ L(X)\geq H_i(X\evec^{(i)}+\left(S^{\Xvec}-X\right)\evec^{(j)}|\G) \bigg\},$$

\medskip

\noindent for all $\Xvec\in \X^n$ and all $\G\in \Sigma$. The latter sets are potentially larger than the former, although taking infima with respect to such sets yield the same values. One aspect that stands out from the previous representation is its individualistic feature. Indeed, each set $D_i(\Xvec,\G)$ depends on the associated agent label $i$. It is intuitive to expect such dependence to disappear under some fairness or anonymity requirement. This justifies the following representation.

\medskip

\begin{theorem}\label{RGRS_prop}
For a mapping $H:\X^n\times \Sigma\to \X^n$, the following are equivalent:

\smallskip

\begin{enumerate}[label=(\roman*)]
    \item $H\in \mathcal{AM}$ and satisfies AA, OA, SI, FP, and ZP.
    \medskip
    \item $H$ is a robust allocation mechanism with $D_i\left(H_i,S,\G\right)=D_j\left(H_j,S,\G\right)$, for all $i,j\in \left[n\right]$, $S\in \X$, and $\G\in \Sigma$.
\end{enumerate}
\end{theorem}

\medskip

\noindent This result highlights how a combination of axioms and properties that we introduced in the previous section fully characterizes robust allocation mechanisms. In particular, it is interesting to note how FP is captured by the \textit{superadditive} nature of each $H_i$. Informally, superadditivity in this setting reflects the idea that it is less costly to the economy if some initial endowments are grouped together in one agent rather than evaluated separately, as less allocation work needs to be performed. OA is captured by the fact that each $H_i$ depends exclusively on $X_i$ and the sum of the endowments $S^{\Xvec}$. SI is expressed through the positive homogeneity of each $H_i$. AA is highlighted by the fact that all supporting sets are equal.

\medskip

In the next section, we will specialize the previous representation in order to obtain a specific functional form for $H$ in terms of conditional expectations.

\medskip
\subsubsection{Robust Conditional Mean Allocation Mechanisms}

We first introduce some notation. Fix $\Xvec\in \X^n$, $\G\in \Sigma$, and a vector $\mathbf{Q}:=\left(Q_1,\ldots,Q_n\right)$ of probability measures over $(\Omega,\F)$. Whenever each $X_i$ is $Q_i$-integrable, we denote by 
\[
\mathbb{E}^{\mathbf{Q}}\left[\Xvec|\mathcal{G}\right]:=
\begin{bmatrix}
\mathbb{E}^{Q_1}\left[X_1|\mathcal{G}\right]\\
\vdots\\
\mathbb{E}^{Q_n}\left[X_n|\mathcal{G}\right]
\end{bmatrix}
\]

\medskip

\noindent a vector of conditional expectations, where each entry identifies the probability measure to be implemented from $\mathbf{Q}$. As in the previous section, we will focus on envelope representations of linear allocation mechanisms. To this end, consider a collection $\left(\mathcal{Q}_i\right)_{i\in \left[n\right]}$ of sets of probability measures over $(\Omega,\F)$, and denote by $\mathcal{Q}$ the product set $\bigtimes_{i\in \left[n\right]}\mathcal{Q}_i$. Analogously, we denote inferior envelopes of vectors of conditional expectations by
\[
\inf\limits_{\mathbf{Q}\in \mathcal{Q}}\mathbb{E}^{\mathbf{Q}}\left[\Xvec|\G\right]:=\begin{bmatrix}  \inf\limits_{Q_1\in \mathcal{Q}_1}\mathbb{E}^{Q_1}\left[X_1|\mathcal{G}\right]\\
\vdots\\
 \inf\limits_{Q_n\in \mathcal{Q}_n}\mathbb{E}^{Q_n}\left[X_n|\mathcal{G}\right]
\end{bmatrix},
\]

\medskip

\noindent for all $\Xvec\in \X^n$ and $\G\in \Sigma$.

\bigskip

\begin{definition}\label{IRGCMRS}
A mapping $H:\X^n\times \Sigma\to \X^n$ is said to be a \textit{robust conditional mean allocation mechanism} if 
    \[
H(\Xvec|\G)=\inf\limits_{\mathbf{Q}\in \mathcal{Q}\left(S^{\Xvec},\G\right)}\mathbb{E}^{\mathbf{Q}}\left[\Xvec|\G^{\Xvec}\right],
    \]

    \medskip

    \noindent     for all $\Xvec\in L^{1}(\F)^n$, where $\mathcal{Q}\left(S^{\Xvec},\G\right):=\bigtimes_{i\in \left[n\right]}\mathcal{Q}_i\left(S^{\Xvec},\G\right)$ and each $\mathcal{Q}_i\left(S^{\Xvec},\G\right)$ is a set of probability measures on $(\Omega,\F)$.
\end{definition}

\medskip

The frictional cost of $H$ is given by
$$\mathcal{C}^{\Xvec|\G}(H) = S^{\Xvec} - \sum_{i=1}^n \inf\limits_{Q_i\in \mathcal{Q}_i\left(S^{\Xvec},\G\right)}\mathbb{E}^{Q_i}\left[X_i|\mathcal{G}^{\Xvec}\right] \geq 0.$$

\vspace{0.1cm}

\noindent Clearly, all robust conditional mean allocation mechanisms are robust allocation mechanisms. The converse is easily seen to be false. Adapting from Denuit et al.\ \cite[Example 7.2]{Denuitetal2022} in the context of  risk sharing rules (see Section \ref{SecRiskSharing} herein),  it is possible to provide an example of a robust allocation mechanism that is not a robust conditional mean allocation mechanism. In particular, an example is provided by the robust allocation mechanism given by

\vspace{-0.1cm}

\[
H_i(\Xvec|\G):=\delta_i X_i+(1-\delta_i)\inf_{Q\in \mathcal{Q}\left(S^{\Xvec},\G\right)}\mathbb{E}^{Q}\left[X_i|\G^{\Xvec}\right],
\]

\medskip

\noindent for all $i\in \left[n\right]$, $\Xvec\in \X^n$, and $\G\in \Sigma$, for some $\delta_i\in \left(0,1\right]$, and for some set of probability measures $\mathcal{Q}(S^{\Xvec},\G)$.

\medskip

Under additional measurability and normalization requirements, we retrieve a sharper representation result, which fully characterizes robust mean allocation mechanisms. Before stating the result we introduce some additional notation. Given $\G\in \Sigma$, we denote by $\bigtriangleup^{\mathbb{P}}(\F|\G)$ the set of $Q\in \bigtriangleup(\F)$ such that $Q\ll \mathbb{P}$ and $Q|_{\G}=\mathbb{P}|_{\G}$. We are now ready to state our main representation result.

\medskip

\begin{theorem}\label{thm1}
Let $H:\X^n\times \Sigma \to \X^n$. The following are equivalent:

\medskip

\begin{enumerate}[label=(\roman*)]
    \item $H\in \mathcal{AM}$ and satisfies the Allocation Axioms and the Allocation Properties.

    \medskip
    
    \item $H$ is a robust conditional mean allocation mechanism with $\mathcal{Q}_i(S,\G)=\mathcal{Q}_j(S,\G)$ and $\mathcal{Q}_i(S,\G)\subseteq \bigtriangleup^{\mathbb{P}}\left(\F|\G^S\right)$, for all $i,j\in \left[n\right]$, $S\in \X$, and $\G\in \Sigma$.
\end{enumerate}
\end{theorem}

\medskip

Theorem \ref{thm1} provides an axiomatization of allocation mechanisms robust to model misspecification. Robust conditional allocation mechanisms determine allocations by minimizing conditional expectations with respect to the least favorable model in each \(\mathcal{Q}_i(S,\G)\). Thus \(\mathcal{Q}_i(S,\G)\) functions as an ambiguity set: lacking full confidence in its knowledge of the data-generating process, the mechanism designer designer considers alternative probabilistic models to hedge against the possibility of employing a misspecified proxy for \(\mathbb{P}\).

\medskip

An interesting difference between Theorem \ref{RGRS_prop} and Theorem \ref{thm1} is that in the latter we do not have to require that $H\in \mathcal{AM}$. This is indeed an implication of FP and IB.

\medskip

\subsubsection{Conditional Mean Allocation Mechanisms}

In this section, we study the properties that lead the set of multiple probabilistic models of a robust conditional mean allocation mechanism to collapse to a singleton. In particular, we introduce the concept of \textit{conditional mean allocation} mechanisms.

\medskip

\begin{definition}
A mapping $H:\X^n\times \Sigma\to \X^n$ is said to be a \textit{conditional mean allocation mechanism} if for all $\Xvec\in \X^n$ and $\G\in \Sigma$, there exists a probability measure $Q^{\G^{\Xvec}}$ such that
$$H_i\left(\Xvec|\G\right) =  \mathbb{E}^{Q^{\G^{\Xvec}}}\left[X_i \vert \G^{\Xvec}\right],$$
for all $i\in \left[n\right]$.
\end{definition}

\medskip

It is immediate to see that any conditional mean allocation mechanism is a robust conditional mean allocation mechanism. However, there are important differences. First, the frictional costs of conditional mean allocation mechanisms are always zero. Indeed, for all $\Xvec\in \X^n$, $\G\in \Sigma$, and $Q\in \bigtriangleup(\F)$, 
\[
S^{\Xvec}-\sum_{i=1}^n\mathbb{E}^{Q}\left[X_i|\G^{\Xvec}\right]=0,
\]

\medskip

\noindent which follows from the linearity of conditional expectations. This \textit{linear} aspect of linear conditional allocation mechanisms is due to the fact that these mechanisms satisfy a condition of \textit{Frictionless Participation}, which we define below.

\medskip

\begin{property}[{\bf Frictionless Participation -- FP*}]
For all $\G \in \Sigma$ and $\Xvec \in \X^n$, if there exists $j \in [n]$ such that $X_j=0$, then for all $i \in [n]$ and all $Z \in \X$,

\vspace{-0.15cm}

$$\mathcal{C}_{i,j}^{\Yvec,\G}(H)= \mathcal{C}_{i,j}^{\Xvec,\G}(H),$$
where
$$\Yvec := \Xvec + Z \evec^{(j)} - Z \evec^{(i)}.$$
\end{property}

\medskip

\noindent This is tantamount to imposing a linearity requirement on $H$. In particular, this property yields the following representation.

\medskip

\begin{proposition}\label{subj_CMRS}
Let $H:\X^n\times \Sigma\to \X^n$. The following are equivalent:

\medskip

\begin{enumerate}[label=(\roman*)]
    \item $H\in \mathcal{AM}$ and satisfies IF, AA, OA, FP*, SI and the Allocation Properties.

    \medskip
    
    \item $H$ is a conditional mean allocation mechanism such that $Q^{\G^{\Xvec}}\in \bigtriangleup^{\mathbb{P}}(\F|\G)$ for all $\G\in \Sigma$ and all $\Xvec\in \X^n$.
\end{enumerate}
\end{proposition}

\bigskip

\section{A Risk-Sharing Perspective}
\label{SecRiskSharing}

Our results for allocation mechanisms are deeply related to the literature on risk sharing. In the latter, endowments are seen as pools of risks that are shared among the agents according to the specificities of an implemented risk-sharing rule. According to this perspective, risk-sharing rules and allocation mechanisms are the same object studied in conceptually different frameworks. The properties of risk-sharing rules are well-studied and involve mainly anonymity and fairness requirements. Recently, Jiao et al.\ \cite{Jiaoetal2023} proposed the first axiomatic study of risk-sharing rules, where they provided a full characterization of the \textit{conditional mean risk sharing rule}, defined as

\vspace{-0.15cm}

$$H\left(\Xvec\right):=\mathbb{E}\left[\Xvec|S^{\Xvec}\right],$$

\medskip

\noindent for all $\Xvec\in \X^n$. In addition, Dhaene et al.\ \cite{Dhaeneetal2023} provide an axiomatization of the so-called \textit{Quantile-Based Risk-Sharing} rule. Our paper contributes to the literature on risk-sharing rules by providing an axiomatic characterization of \textit{robust risk sharing rules}, the risk-sharing analogues of our robust allocation mechanisms.

\medskip

In this setting, we use the term \textit{generalized risk sharing rule} to refer to allocation mechanisms. Following Denuit et al.\ \cite{Denuitetal2022}, we highlight some of the properties satisfied by a robust risk sharing rule. Thanks to the previous equivalence results, some of the properties are an immediate consequence of our axioms.

\medskip

\begin{definition}\label{Properties_def}
We say that a generalized risk sharing rule $H$ satisfies:

\medskip

\begin{enumerate}
    \item \textit{Positive Homogeneity} if $H(\alpha \Xvec|\cdot)=\alpha H(\Xvec|\cdot)$ for all $\alpha>0$ and $\Xvec\in \X^n$.

    \medskip
    
    \item \textit{Normalization} if $H_j(\Xvec|\cdot)=0$ for all $\Xvec\in \X^n$ with $X_j=0$ and $j\in \left[n\right]$.

    \medskip
    
    \item \textit{Constancy}  if $H_j(\Xvec|\cdot)=c$ for all $\Xvec\in \X^n$ with $X_j=c$, $c\in \mathbb{R}$, and $j\in \left[n\right]$.

    \medskip
    
    \item \textit{Translativity} if $H_j(\Xvec+c\evec^{(j)}|\cdot)=H_j(\Xvec|\cdot)+c$ for all $\Xvec\in \X^n$, $c\in \mathbb{R}$, and $j\in \left[n\right]$.

    \medskip
    
    \item \textit{No-Ripoff} if $H_j(\Xvec|\cdot)\leq \sup X_j$ for all $\Xvec\in \X^n$ and $j\in \left[n\right]$. 

    \medskip
    
    \item \textit{Actuarial Fairness} if $\mathbb{E}\left[H_i(\Xvec|\cdot)\right]=\mathbb{E}\left[X_i\right]$ for all $i\in \left[n\right]$.
\end{enumerate}
\end{definition}

\medskip

Analogously, we refer to robust conditional mean allocation mechanism as \textit{robust conditional mean risk sharing rule} in the risk sharing context. 

\medskip

\begin{proposition}\label{prop_properties}
Suppose that $H$ is a robust conditional mean risk sharing rule with $\mathcal{Q}_i(S,\G)=\mathcal{Q}_j(S,\G)$ and $\mathcal{Q}_i(S,\G)\subseteq \bigtriangleup^{\mathbb{P}}\left(\F|\G^S\right)$ for all $i,j\in \left[n\right]$, $S\in \X$, and $\G\in \Sigma$. Then it satisfies Positive Homogeneity, Normalization, Constancy, Translativity, and No-Ripoff.
\end{proposition}

\medskip

It is important to point out that not all robust conditional mean risk sharing rules satisfy Actuarial Fairness. This is quite easily observed. In particular, suppose that $H\in \mathcal{AM}$ satisfies Actuarial Fairness. Then for all $\Xvec\in \mathcal{X}^n$ and $\G\in \Sigma$, 
\[
\sum_{i=1}^nH_i(\Xvec|\G)\leq S^{\Xvec}\ \textnormal{and}\ \sum_{i=1}^n\mathbb{E}\left[H_i(\Xvec|\G)\right]=\mathbb{E}\left[S^{\Xvec}\right].
\]
This implies that $\sum_{i=1}^nH_i(\Xvec|\G)= S^{\Xvec}$. By a contrapositive argument, it readily follows that if $H\in \mathcal{AM}$ and $\sum_{i=1}^nH_i(\Xvec|\G)|_A< S^{\Xvec}|_A$ for some $A\in \F$ with $\mathbb{P}\left[A\right]>0$, then $H$ cannot satisfy Actuarial Fairness.

\medskip

Proposition \ref{subj_CMRS} allows us to provide an axiomatization of what can be seen as a subjective version of the CMRS.

\begin{definition}
A generalized risk-sharing rule $H$ is a \textit{subjective conditional mean-risk sharing rule} if for all $\Xvec\in \X^n$ and $\G\in \Sigma$, there exists a probability measure $Q^{\G^{\Xvec}}$ such that

\vspace{-0.15cm}

$$H_i\left(\Xvec|\G\right) =  \mathbb{E}^{Q^{\G^{\Xvec}}}\left[X_i \vert \G^{\Xvec}\right], \ \hbox{for all } i\in \left[n\right].$$
\end{definition}

\medskip

\noindent It then follows from Proposition \ref{subj_CMRS} that a generalized risk-sharing rule $H$ satisfies IF, AA, OA, FP*, SI and the Allocation Properties if and only if it is a subjective conditional mean-risk sharing rule, with $Q^{\G^{\Xvec}}\in \bigtriangleup^{\mathbb{P}}(\F|\G^{\Xvec})$, for all $\G\in \Sigma$ and all $\Xvec\in \X^n$. 

\bigskip

\section{Two Examples}
\label{SecApplication}
In this section we discuss two examples of robust allocation mechanisms. In each, we propose a model of market allocation based on a given robust allocation mechanism that allows the global frictional cost in the market to be parameterized by a single parameter. This parameter can be understood as related to the fees imposed by the allocator on the participants.

\medskip
\subsection{Mean-Deviation Allocation Mechanisms}
\label{example:meandeviation}

First we consider a mean-deviation approach, based on the theory of generalized deviation measures of Rockafellar et al.\ \cite{RockafellarUryasev06}.

\medskip

\begin{definition}[Mean-Deviation]
Let $D:\X\times \Sigma\to \X$ be such that:
\medskip
\begin{itemize}
\item For all $\G\in \Sigma$, $D(\cdot|\G)$ is subadditive and positively homogeneous.
\medskip
\item For all $\G\in\Sigma$, $\G$-measurable $Y\in \X$, and $X\in \X$, $D(X+Y|\G)=D(X|\G)$.
\medskip
\item For all $\G\in \Sigma$ and $X\in \X$, $D(X|\G)$ is $\G$-measurable and
\[
X\geq 0\Longrightarrow D(X|\G)\leq \mathbb{E}\left[X|\G\right].
\]
\item For all $\G\in \Sigma$, $D(\cdot|\G)\geq 0$.
\end{itemize}
\end{definition}

\bigskip

\noindent Given $\Xvec\in \X^n$ and $\G\in \Sigma$, $D(\cdot|\G^{\Xvec})$ can be seen as a generalized conditional volatility measure (see \cite{RockafellarUryasev06} for more on this). For a given $\theta\in [0,\infty)$, we define the \textit{mean-deviation} allocation mechanism $H^{D,\theta}$ by
\[
H_{i}^{D,\theta}(\Xvec|\G)=\mathbb{E}\left[X_i|\G^{\Xvec}\right]-\theta \, D(X_i|\G^{\Xvec}),
\]

\medskip

\noindent for all $i\in [n]$, $\Xvec \in \X^n$, and $\G\in \Sigma$. According to this allocation mechanism, each participant receives the conditional expectation of their initial endowment, minus a cost that is a multiple $\theta$ of the generalized conditional volatility of their initial endowment. The parameter $\theta$ is controlled by the pooling system, which can vary the parameter (or fee, or loading factor) to control the number of individuals in the pool. 

\medskip

It is straightforward to verify that, when it satisfies CA, $H^{D,\theta}$ is a robust conditional mean allocation mechanism. The global frictional cost of such an allocation mechanism is given by
\[
 \mathcal{C}^{\Xvec,\G}(H^{D,\theta})=\sum_{i=1}^n \theta \, D(X_i|\G^{\Xvec}).
\]
Therefore, the larger the deviations or the loading factor $\theta$, the higher the cost.

\bigskip
\subsection{Expected-Shortfall Allocation Mechanisms}
\label{example:ES}

The second example we consider is inspired by one of the most popular and most studied measures of risk: the \textit{Expected Shortfall}. 

\medskip

\begin{definition}[Left-Tail Expected Shortfall]
Let $\lambda\in (0,1)$, and define 
\[
\mathcal{Q}^{\lambda}(S,\G):=\left\lbrace (Q_i)_{i=1}^n\in \bigtriangleup^{\mathbb{P}}(\F|\G^S)^n:\forall i\in \left[n\right],\ \frac{\textnormal{d}Q_i}{\textnormal{d}\mathbb{P}}\leq \frac{1}{\lambda} \right\rbrace.
\]

\noindent We refer to the robust allocation mechanism induced by $\mathcal{Q}^{\lambda}(S,\G)$ as a  \textit{Left Expected Shortfall} (LES) mechanism. In particular, this mechanism is given by
\[
H^{\lambda}(\Xvec|\G)=\inf_{\mathbf{Q}\in \mathcal{Q}^{\lambda}(S^{\Xvec},\G)} \, \mathbb{E}^{\mathbf{Q}}\left[\Xvec|\G^{\Xvec}\right],
\]

\medskip

\noindent for all $\Xvec\in \X^n$ and $\G\in \Sigma$. 
\end{definition}

\medskip

The sets $\mathcal{Q}^{\lambda}(S,\G)$ are again interpreted as ambiguity sets employed by the mechanism designer to hedge against model misspecification. In this case, the alternative probabilistic models considered by the mechanism designer cannot be too extreme, as their density (w.r.t.\ $\mathbb{P}$) is capped by $1/\lambda$. Therefore, $\lambda$ can be interpreted as a factor of confidence of the mechanism designer about its understanding of $\mathbb{P}$: the closer $\lambda$ is to $0$, the more cautious the approach of the designer is.

\medskip

This abstract specification turns out to be handy when it comes to the analysis of allocations with some specific distributions of the initial endowments. As an illustration, we consider the case of no initial information (i.e., $\G=\left\lbrace \emptyset,\Omega\right\rbrace$) and a vector of normally distributed endowments $\Xvec=(X_1,\ldots,X_n)\sim \mathcal{N}(\boldsymbol{\mu},\mathbf{V})$, with
\[
\boldsymbol{\mu}=(\mu_1,\ldots,\mu_n)\ \ \textnormal{and}\ \ \mathbf{V}_{ij}=\rho_{ij}\sigma_i\sigma_j,
\]

\medskip

\noindent where $\rho_{ij}$ is the correlation coefficient between $X_i$ and $X_j$, for all $i,j\in [n]$. Formulas for computing the Expected Shortfall under normality can be found in e.g., \cite{mcneil2015quantitative}; see also details in Appendix \ref{app:ex52}. For simplicity of notation, let 
\[
\sigma_{1:n}^2:=\sum_{i=1}^n\sigma^2_i+2\sum_{i<j}\rho_{ij}\sigma_i\sigma_j\ \ \textnormal{ and } \ \
\bar{\rho}_i:=\frac{\sum_{j=1}^n\rho_{ij}\sigma_i\sigma_j}{\sigma_i\sigma_{1:n}}\ \textnormal{for\ all}\ i\in [n].
\]

\medskip

\noindent In this case, for all $i\in [n]$, 
\[
H_i^{\lambda}(\Xvec|\G)=\mu_i+\frac{\sigma_i}{\sigma_{1:n}}\bar{\rho}_i\left(S^{\Xvec}-\sum_{i=1}^n\mu_i\right)-\frac{\varphi(\Phi^{-1}(\lambda))}{\lambda}\sqrt{(1-\bar{\rho}_i^2)\sigma^2_i},
\]
and the global frictional cost is
\[
\mathcal{C}^{\Xvec,\G}(H)=\frac{\varphi\left(\Phi^{-1}(\lambda)\right)}{\lambda}\sum_{i=1}^n\sqrt{(1-\bar{\rho}_i^2)\sigma^2_i},
\]
where $\varphi$ and $\Phi$ are the density and CDF of the standard normal distribution, respectively. 
Note that $H_i^\lambda(\mathbf X|\mathcal G)$ is again normally distributed. 
These formulas allow us to analyze the role of volatility and correlation of the endowments on the frictions/benefits suffered/enjoyed by the participants from risk-sharing. First, Figure \ref{fig1:costsvscorrelation} below illustrates the impact of correlation on the global frictional cost.

\begin{figure}[H]
\includegraphics[width=10cm]
{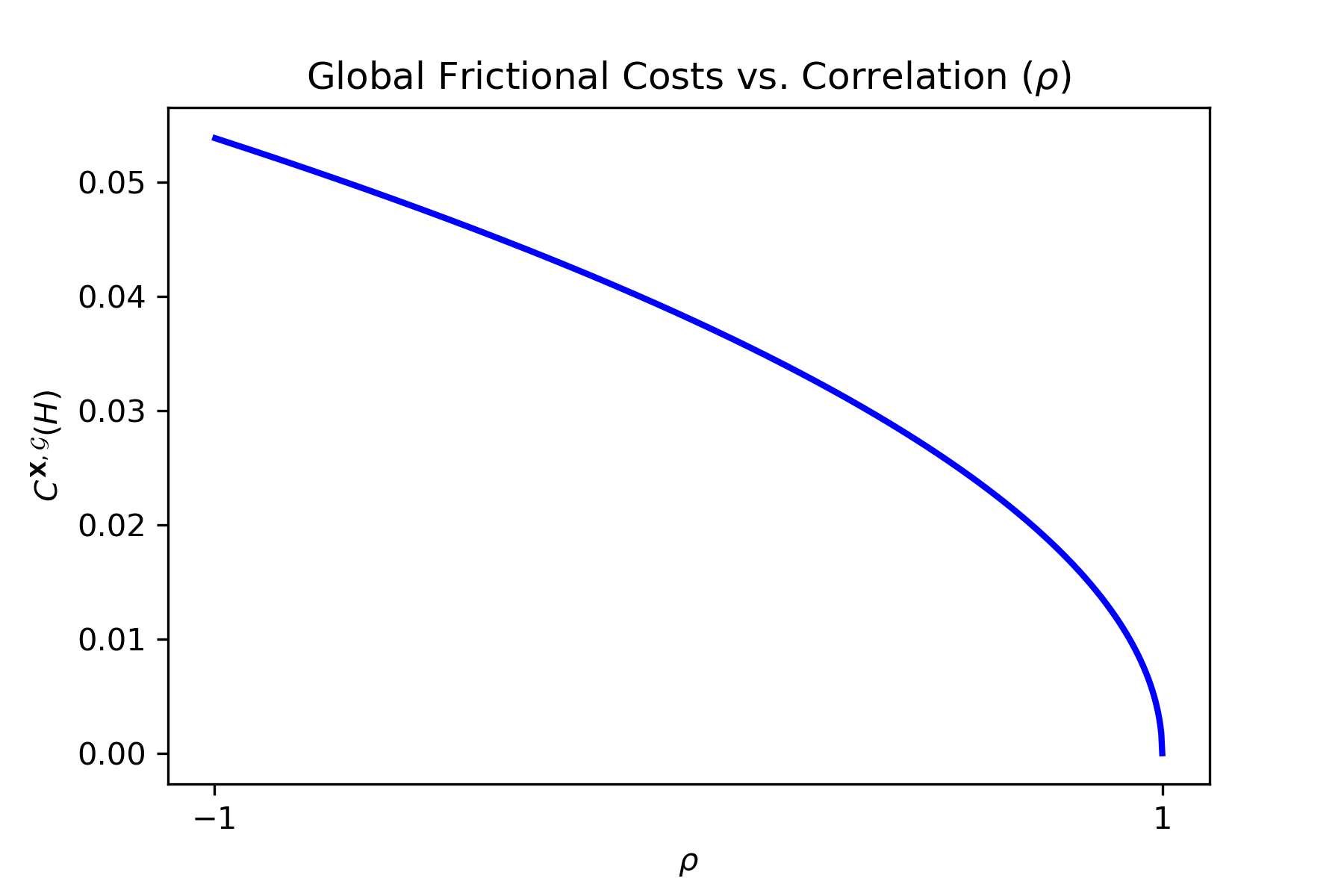}
\caption{ $\sigma_1=\sigma_2=1$, $\lambda=0.99$, $\rho_{ij}=\rho$ for all $i\neq j$.}
\label{fig1:costsvscorrelation}
\end{figure}

\medskip

Interpreting the global frictional cost as the \textit{profits} of the allocator, Figure \ref{fig1:costsvscorrelation} highlights how these depend negatively and concavely on the level of correlation, which we assume constant in this example, that is, $\rho_{ij}=\rho$, for some $\rho \in [-1,1]$, and for all $i\neq j$. The higher the correlation the more limited are the possibilities and the efficiency of re-allocating the starting endowments to rebalance the risks, and hence this motivates a lower allocation charge to the participants. A different interpretation is that a higher level of correlation is statistically equivalent to less independent components, and consequently a lower amount of data to be processed, thereby resulting in less frictions.

\medskip

We now take this analysis a step further and assume, in addition, that the participants have mean-variance preferences, i.e., their utility functions are given by
\[
V_i:X \mapsto \mathbb{E}[X]-\theta_i \,  \sigma^2_X,
\]
where $\theta_i$ is a risk attitude parameter. 
For normal random variables (recall that   the initial endowments and the allocation  are normally distributed in our context), the mean-variance preferences can be equivalently described by expected utilities using exponential utility functions. 

\medskip

Considering the preferences of the individuals allow us to evaluate the trade-off they face and their incentive to participate in the risk-sharing pool. This trade-off is simply expressed by the following difference:
\[
T_i(\Xvec|\G)=V_i(H_i(\Xvec|\G))-V_i(X_i).
\]

\smallskip

\noindent In particular, whenever $T_i(\Xvec|\G)\geq 0$, agent $i$ wishes to enter the risk-sharing pool, as $i$ prefers the re-allocated $H_i(\Xvec|\G)$ to his starting endowment $X_i$. Requiring $T_i\geq 0$ can be seen as an individual rationality (IR) constraint. In the case that we consider here, with normally distributed endowments and no initial information, this takes the form
\[
T_i(\Xvec|\G)=\sqrt{(1-\bar{\rho}^2_i)\sigma^2_i\,} \, \left[\theta_i\sqrt{(1-\bar{\rho}^2_i)\sigma^2_i\,}-\frac{\varphi(\Phi^{-1}(\lambda))}{\lambda}\right].
\]

\medskip

\noindent The volatility, correlations, levels of risk aversion, and $\lambda$'s all play a role in determining a threshold for agent $i$ (namely, $T_i(\Xvec|\G) \geq 0$) beyond which agent $i$ has an incentive to enter the risk-sharing pool. A natural question to examine is how to evaluate such trade-offs and find regions that incentivize individuals to participate in the pool and the allocator to implement the pooling mechanism.

\medskip

First, it is immediately noticeable that for all $i\in [n]$, whenever $\theta_i\leq 0$, individual $i$ has no incentive to enter the pool. This is quite intuititve, as risk-neutral (and risk-loving) agents have no interest in diversification. Therefore, in the following figures we consider only risk averse agents.

\begin{figure}[H]
\includegraphics [scale=0.8]{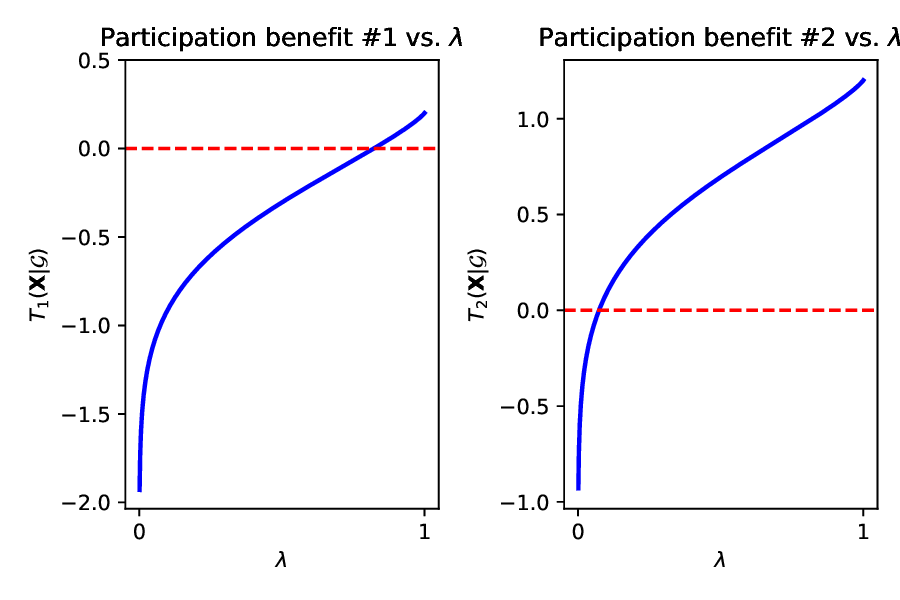}
\caption{$n=2$, $\sigma_1=\sigma_2=1$, $\theta_1=1/2$, $\theta_2=3$, $\rho_{ij}=\frac{2}{10}$ for all $i\neq j$.}
\label{fig3:costsvsnumberofpart}
\end{figure}

\noindent As Figure \ref{fig3:costsvsnumberofpart} shows,  the minimum value $\lambda^*$, which determines the threshold of participation for the individuals,\footnote{More formally, $\lambda_i^*(\theta)=\inf\left\lbrace \lambda\in (0,1]:\sqrt{(1-\bar{\rho}_i^2)\sigma_i^2\,} \, \left[\theta\sqrt{(1-\bar{\rho}_i^2)\sigma_i^2\,}-\frac{\varphi(\Phi^{-1}(\lambda))}{\lambda}\right]=0 \right\rbrace$ for all $\theta> 0$ and $i\in [n]$. In the context of Figure \ref{fig3:costsvsnumberofpart} we omit the subscript $i$ from $\lambda^*_i$, as $\sigma_1=\sigma_2$ and $\rho_{ij}$ is constant.} depends significantly on the level of risk aversion. Indeed, in the figure, $\lambda^*(\theta_1)>\lambda^*(\theta_2)$, highlighting how more risk aversion leads to a higher willingness to pay and enter the risk-sharing pool. 

\medskip

Next we analyze how the average trade-offs and average global frictional cost vary with the number of participants in the pool. For both, we consider two cases: varying and fixed correlations.

\begin{figure}[H]
\begin{subfigure}[H]{0.45\linewidth}
\includegraphics[width=7.5cm]{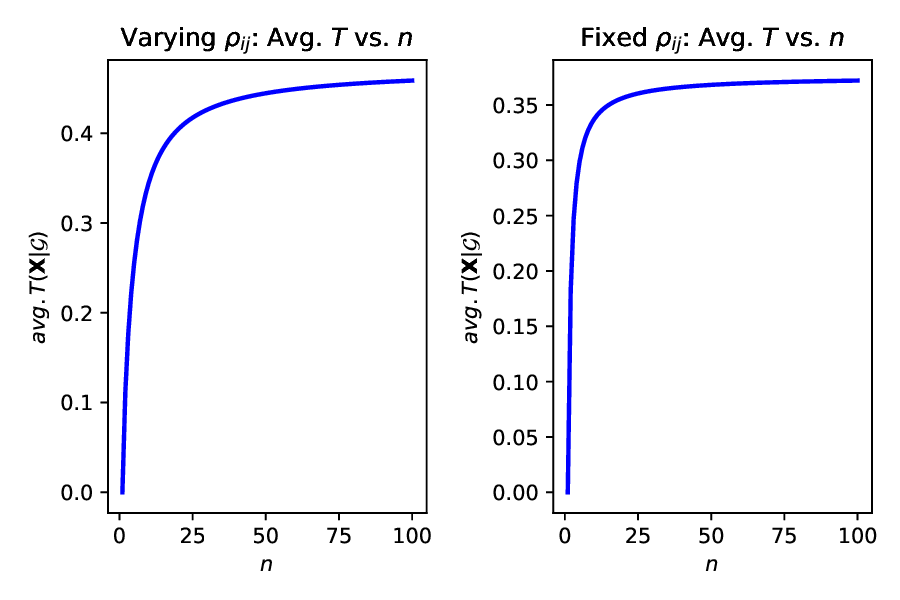}
\caption{Average participation benefits}
\end{subfigure}
\hspace{0.2cm}
\begin{subfigure}[H]{0.45\linewidth}
\includegraphics[width=7.5cm]{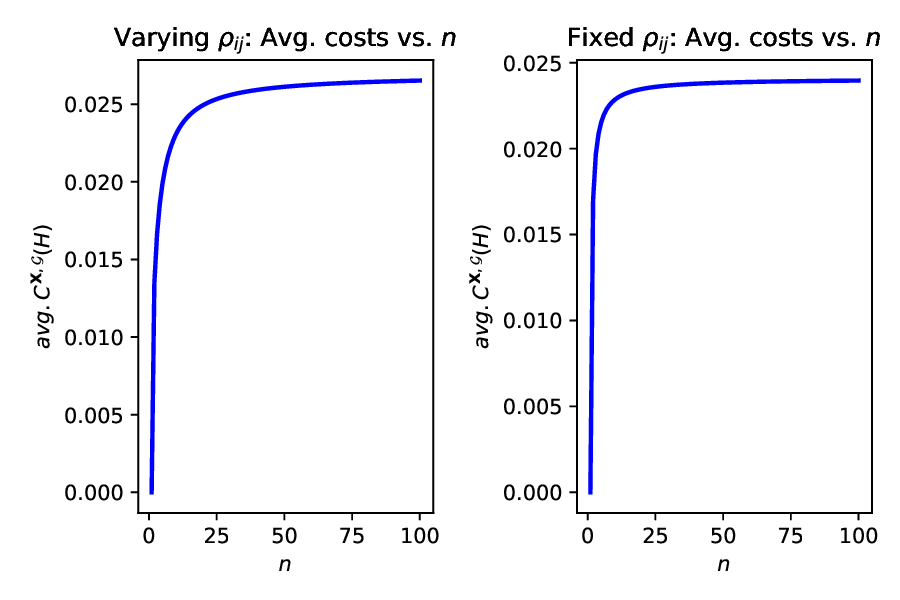}
\caption{Average global frictional costs}
\end{subfigure}%
\caption{Varying $\rho_{ij}=\left(\frac{1}{2}\right)^{|i-j|}$, fixed $\rho_{ij}=\frac{2}{10}$, $\sigma_i=1$, $\theta_i=\frac{1}{2}$, $\lambda =0.99$.}
\label{figure:averagevalues}
\end{figure}

Figure \ref{figure:averagevalues} illustrates two important facts. First, the IR constraint for each member of the pool (i.e., $T_i\geq 0$) is easier to satisfy when the number of participants is higher. A possible interpretation is that the higher the number of participants, the stronger the possible benefits from diversification, and the more attractive is the allocation mechanism to possible new participants. Second, the average global frictional cost increases with the number of participants. An interpretation is that a higher number of participants makes the allocation process more burdensome (e.g., more data to be processed), leading to higher frictions,
translating to more profit on average for the provider of the allocation service. 

\medskip

\subsection{Empirical Implementation with Flood Loss Data}

We provide an illustrative example using real-world data on insurance claims related to flood losses in the U.S. The data on historical insurance claims  are collected by the U.S.~Federal Emergency Management Agency (FEMA), under the National Flood Insurance Program (NFIP). The dataset used is publicly available at the OpenFEMA data platform. We consider two sets of three states: California--New York--Texas and Alabama--Lousiana--Mississippi that are distinct for the geographical distance of their members.

\begin{table}[H]
\centering
\label{tab:flood_stats_and_corr1}
\begin{minipage}{0.55\textwidth}
\centering
\textbf{Summary Statistics} \\[0.5em]
\begin{tabular}{lrrr}
\toprule
\textbf{Statistic} & \textbf{CA} & \textbf{NY} & \textbf{TX} \\
\midrule
Mean & $-16.04\times 10^5$ & $-11\times 10^6$ & $-33.25\times 10^6$ \\
Variance & $7.51\times 10^{13}$ & $3.47\times 10^{16}$ & $1.77\times 10^{17}$ \\
\bottomrule
\end{tabular}
\end{minipage}%
\hfill
\begin{minipage}{0.40\textwidth}
\centering
\textbf{Correlation Matrix} \\[0.5em]
\begin{tabular}{lrrr}
\toprule
 & \textbf{CA} & \textbf{NY} & \textbf{TX} \\
\midrule
\textbf{CA} & 1 & -0.094 & -0.109 \\
\textbf{NY} & -0.094 & 1 & -0.044 \\
\textbf{TX} & -0.109 & -0.044 & 1 \\
\bottomrule
\end{tabular}
\end{minipage}
\caption{Summary statistics and correlation matrix for CA--NY--TX flood losses.}
\end{table}

\medskip

\begin{table}[H]
\centering
\label{tab:flood_stats_and_corr2}
\begin{minipage}{0.55\textwidth}
\centering
\textbf{Summary Statistics} \\[0.5em]
\begin{tabular}{lrrr}
\toprule
\textbf{Statistic} & \textbf{AL} & \textbf{LA} & \textbf{MS} \\
\midrule
Mean     & $-23.96\times 10^5$ & $-44.74\times 10^6$ & $-78.68\times 10^5$ \\
Variance & $5.15\times 10^{14}$ & $5.09\times 10^{17}$ & $1.88\times 10^{16}$ \\
\bottomrule
\end{tabular}
\end{minipage}%
\hfill
\begin{minipage}{0.40\textwidth}
\centering
\textbf{Correlation Matrix} \\[0.5em]
\begin{tabular}{lrrr}
\toprule
 & \textbf{AL} & \textbf{LA} & \textbf{MS} \\
\midrule
\textbf{AL} & 1     & 0.596 & 0.606 \\
\textbf{LA} & 0.596 & 1     & 0.983 \\
\textbf{MS} & 0.606 & 0.983 & 1     \\
\bottomrule
\end{tabular}
\end{minipage}
\caption{Summary statistics and correlation matrix for AL--LA--MS flood losses.}
\end{table}

\smallskip

We assume that the data $X_{\textbf{CA}}$, $X_{\textbf{NY}}$, $X_{\textbf{TX}}$,  $X_{\textbf{AL}}$, $X_{\textbf{LA}}$, and $X_{\textbf{MS}}$ are drawn from normal distributions.\footnote{This is not the best fit, but it serves purpose to apply the expected shortfall mechanism presented above to real-world data. More comprehensive empirical analyses are left for future research.}  This simplifying assumption allows us to apply the formulas retrieved in the previous section to compute the allocations, costs, and benefits obtained through the implementation of the Expected Shortfall Allocation Mechanism to the two groups of states.

\medskip

\begin{table}[ht]
\centering
\renewcommand{\arraystretch}{1.3}
\begin{tabular}{lccccc}
\toprule
\textbf{State} & $\mathbb{E}\left[H_i^{\lambda}(\Xvec|\G)\right]$ & $\mathbb{E}[C_i^{\Xvec,\G}(H)]$ & $T_i(\Xvec|\G)$ \\
\midrule
\textbf{CA}  &  $-18.37\times 10^5$ &  $23.32\times 10^4$ & $3.753\times 10^{13}$ \\
\textbf{NY} & $-15.60\times 10^6$ & $45.93\times 10^5$ & $1.455\times 10^{16}$ \\
\textbf{TX} & $-37.85\times 10^6$ & $45.96\times 10^5$ & $1.457\times 10^{16}$ \\
\bottomrule
\end{tabular}
\caption{Allocations CA, NY, and TX flood losses with $\lambda=0.99$ and $\theta_i\equiv 1/2$.}
\label{tab:summary_stats_CANYTX}
\end{table}

\begin{table}[ht]
\centering
\renewcommand{\arraystretch}{1.3}
\begin{tabular}{lccccc}
\toprule
\textbf{State} & $\mathbb{E}\left[H_i^{\lambda}(\Xvec|\G)\right]$ & $\mathbb{E}[C_i^{\Xvec,\G}(H)]$ & $T_i(\Xvec|\G)$ \\
\midrule
\textbf{AL} & $-28.77\times 10^5$ &  $48.18\times 10^4$ & $1.601\times 10^{14}$ \\
\textbf{LA} & $-45.48\times 10^6$ & $74.03\times 10^4$ & $3.781\times 10^{14}$ \\
\textbf{MS} & $-84.40\times 10^5$ & $57.15\times 10^4$ & $2.253\times 10^{14}$ \\
\bottomrule
\end{tabular}
\caption{Allocations AL, LA, and MS flood losses with $\lambda=0.99$ and $\theta_i\equiv 1/2$.}
\label{tab:summary_stats_ALLAMS}
\end{table}

\medskip

This empirical illustration implements the Expected Shortfall Allocation Mechanism on NFIP flood-claim data for two triads (CA--NY--TX and AL--LA--MS), chosen to contrast weakly and strongly correlated risk pools. Thanks to the Gaussian distribution hypothesis adopted here, the mechanism yields closed-form allocations and costs. Consistently with the formulas derived in Section~\ref{example:ES}, within each triad, the state with the largest variance bears both the highest expected cost share and the greatest participation gain, while higher cross-state dependence compresses diversification benefits and thus lowers the designer profits. Interpreting the global frictional cost as a platform fee, the exercise shows how tuning $\lambda$ trades off larger frictions against inclusion (higher participation benefits), and how pool composition (geographical location affects correlation) shapes both costs and incentives.

\bigskip

\section{A Discussion of the Mathematical Contribution}
\label{MathLit}

The results presented here on the axiomatization of allocation mechanisms are based on applications of Hahn-Banach type-extension results (Theorems \ref{HBgen} and \ref{HBKgen}). These results provide conditions for extending linear operators on ordered vector spaces. We study some implications of these extension results for envelope representations of sublinear operators. As in the standard real-valued setting, we show that sublinear operators can be written as the upper envelopes of their dominated linear operators. Our results abstract from any topological assumptions, and they are purely order-theoretic and algebraic in nature. The level of generality adopted makes our results applicable to all order-theoretic contexts where sublinear operators play a role. In particular, we retrieve such envelope representation results within Dedekind complete Riesz spaces, that include a large amount of widely used ordered vector spaces. The provided representation results are then refined, adding further restrictions on the behavior of the sublinear operators, and they highlight a classic feature: the dominated linear operators often inherit the same properties of the dominating sublinear operator. Among such restrictions, the main ones are monotonicity and normalization, that are both meaningful in economics and finance. Monotonicity is a non-satiation assumption of the type ``\textit{more is better}'', while normalization says that the value of certain outcomes should remain unchanged after any evaluation. This motivates the mathematical analysis of such properties, that is provided in the \hyperlink{LinkToAppendix}{Appendices}. To the best of our knowledge, the generalization of our results toward convex operators remains an open question.

\medskip

From a mathematical perspective, our representations results take place in the literature on convex duality. In particular, we provide envelope representation results for superlinear operators mapping from Dedekind complete Riesz spaces to Dedekind complete Riesz spaces. These results are then applied in the context of finite Cartesian products of $L^1$-spaces and conditional expectation representations. From this perspective the closest work to ours are the studies of conditional risk measures by Filipovi\'{c} et al.\ \cite{FilKupVogel2012}, Detlefsen and Scandolo \cite{DetlefsenScandolo2005}, as well as Bion-Nadal and Palaiesau \cite{BionNadal04}, for instance. In the former, the authors provide a comprehensive analysis of two approaches to the study of conditional risk measures. The first approach relies on classical convex analysis, treating risk measures as functions on vector spaces; whereas the second approach builds instead upon the theory of random modules. Our analysis is closer to the first approach, with the difference that we do not rely on any explicit local property to provide a conditional expectation representation (see for instance their Proposition 2.5). The work of Detlefsen and Scandolo \cite{DetlefsenScandolo2005} and Bion-Nadal and Palaiesau \cite{BionNadal04} are similar to the first approach of Filipovi\'{c} et al.\ \cite{FilKupVogel2012}. A further relevant difference of our paper is that even if we are able in the proofs to reduce the dimensionality of our problem, in the end we provide an envelope representation result for a non-superlinear operator. In addition to the mentioned papers, there are many other contributions to the theory of conditional risk measures that we do not mention for the sake of brevity. On the side of convex analytic representations, Bartl \cite{Bartl2020} provided representation results for conditional nonlinear expectations. While the operators we study are very similar, the approach in Bartl \cite{Bartl2020} is oriented towards \textit{state-dependent} representations under additional topological assumptions. Our results are more \textit{functional-analytically and order-theoretically} oriented than those of Bartl \cite{Bartl2020}. Moreover, we provide representations in terms of linear operators representable as conditional expectations, while Bartl \cite{Bartl2020} focuses on state-dependent representations that allow to work with linear operators, representable as unconditional expectations. The conditional aspects in Bartl \cite{Bartl2020} are captured by the dependence of the set of multiple probabilistic models on each state the world $\omega \in \Omega$.

\bigskip

\section{Conclusion}
\label{SecConc}

In this paper, we contribute to the literature on efficient allocations in pure-exchange economies by presenting an axiomatic study of allocation mechanisms in the presence of frictions on certain types of initial transfers of endowments. Such mechanisms are mappings that transform feasible re-allocations of the aggregate endowment into other feasible re-allocations.

\medskip

We argue that in a pure-exchange economy, certain kinds of subsidization, namely initial transfers of state-contingent endowments to any agent with zero initial endowment, are ultimately costly to the economy. We capture the resulting frictional costs through an axiom that we call \textit{Frictional Participation}, and we show that the existence of these frictional costs in the economy is equivalent to a form of subadditivity of the cost of transferring endowments between agents in the economy. We illustrate the far-reaching consequences of \textit{Frictional Participation} by providing, among other results, an axiomatic characterization of those allocation mechanisms that admit representations as robust (worst-case) linear allocation mechanisms, as well as those mechanisms that admit representations as worst-case conditional expectations. We call the latter \textit{Robust Conditional Mean Allocation} mechanisms, and in the context of risk-sharing within a pool of agents we refer to them as \textit{Robust Conditional Mean Risk Sharing} rules. Those risk sharing rules are of particular importance, as they provide an extension of the so-called \textit{Conditional Mean-Risk Sharing} rule of Denuit and Dhaene \cite{DenuitDhaene2012} and the \textit{Generalized Conditional Mean-Risk Sharing} rule of Jiao et al.\ \cite{Jiaoetal2023}. As a corollary, we provide an axiomatization of (linear) \textit{Conditional Mean Allocation} mechanisms, as well as a subjective version of the \textit{Conditional Mean-Risk Sharing} rule. 

\medskip

Some of our results can offer ways to design robust decentralized systems of trade, such as decentralized finance risk-sharing protocols. As an illustration, we discuss two interesting and practically relevant examples of robust allocation mechanisms: the Mean-Deviation allocation mechanism and the Expected-Shortfall allocation mechanism. In each example, we propose a model of market allocation based on the given robust allocation mechanism, and we show how the global frictional cost in the market can be parameterized by a single parameter. This parameter can be understood as related to the fees imposed by the allocator on the participants. In particular, the Expected Shortfall Allocation Mechanism presented in Section \ref{example:ES}, by choosing $\mathcal G=\{\emptyset,\Omega\}$, yields an explicit allocation rule that can be implemented under any statistical specification of the vector of initial endowments. This mechanism features a free parameter $\lambda$, which regulates the extent of friction in the risk-sharing arrangement and represents the random fees collected by the platform.

\medskip

On a technical level, our main results require novel envelope representation results for superlinear operators mapping from Dedekind complete Riesz spaces to Dedekind complete Riesz spaces. To the best of our knowledge, such results were absent from the mathematical literature on convex duality, and we fill this gap herein.

\newpage

\setlength{\parskip}{0.5ex}
\hypertarget{LinkToAppendix}{\ }
\appendix

\vspace{-0.8cm}

\section{Order-Theoretic Background}
\label{AppBackground}

\subsection{Definitions}
Given a set $V$, we say that a binary relation $\succeq$ is a \textit{partial order} if it is reflexive, transitive, and antisymmetric. A partially ordered set $(V,\succeq)$ will be referred to as a \textit{poset}. A poset $(V,\succeq)$ is said to be a \textit{lattice} if $\inf_{\succeq}\left\lbrace x,y\right\rbrace\in V$ and $\sup_{\succeq}\left\lbrace x,y\right\rbrace\in V$, for all $x,y\in V$. A subset $B$ of a poset $V$ is $\succeq$-\textit{bounded} if there exists $x\in V$ such that $x\succeq y$, for all $y\in B$. Moreover, we say that a poset $(V,\succeq)$ is \textit{Dedekind complete} if $\sup_{\succeq} B$ exists in $V$, for all nonempty $\succeq$-bounded subsets $B\subseteq V$. In words, a poset $(V,\succeq)$ is Dedekind complete if every nonempty subset of $V$ with an $\succeq$-upper bound in $V$ has a $\succeq$-supremum in $V$. We say that a real vector space $(V,\succeq)$ that is also a poset is a \textit{partially ordered vector space} if $$x\succeq y\Longleftrightarrow x+z\succeq y+z\ \textnormal{and}\ x\succeq y\Longleftrightarrow \alpha x\succeq \alpha y,$$
for all $x,y,z\in V$ and all $\alpha> 0$. A partially ordered vector space $(V,\succeq)$ is said to be a \textit{Riesz space} if it is also a lattice. In what follows, whenever the order is understood from the context we will write $\sup$ and $\inf$ in place of $\sup_{\succeq}$ and $\inf_{\succeq}$. A subset $S$ of a Riesz space $(V,\succeq )$ is called \textit{solid} if $|x|\succeq |y|$ and $x\in S$ imply that $y\in S$. A solid vector subspace of a Riesz space is called an \textit{ideal}. In what follows we will make use of the following lemma.

\medskip

\begin{lemma}
\label{ideals}
An ideal in a Dedekind complete Riesz space is Dedekind complete.
\end{lemma}
\begin{proof}
See \cite[Lemma 8.14]{AliprantisBorder}.
\end{proof}

\medskip

\subsection{Lattice and Pointwise Suprema in Lebesgue Spaces}
Fix a finite measure space $(\Omega, \mathcal{F},\mu)$, it is well-known that pointwise suprema of uncountably many measurable functions may not be measurable. However, this is not the case for \textit{lattice suprema}. Before stating this classical result, we report here some terminology. Denote by $L^0(\F)$ the space of $\F$-measurable and real-valued functions, which we endow with the $\mu$-a.e.\ dominance order, denoted by $\geq_{\mu\textnormal{-a.e.}}$, as customary $\mu$-a.e. equal measurable functions are identified as equal. For $F\subseteq L^0(\F)$, we denote by $\sup F$ the \textit{pointwise supremum} of $F$, which is defined as
\[
\sup F:\omega\mapsto \sup\left\lbrace f(\omega):f\in F \right\rbrace.
\]

\medskip

\noindent If $F$ is countable, then $\sup F$ is $\F$-measurable. The \textit{lattice supremum} of $F$ is defined as $\sup_{\geq_{\mu\textnormal{-a.e.}}}F$. When $F$ is countable, its pointwise and lattice suprema coincide. This is not true in general, when $F$ is uncountable. In order, to bypass this issue, we will employ the countable sup property of spaces of measurable functions. A Dedekind complete poset $(V,\succeq)$ satisfies the \textit{countable sup property} if, for all order bounded sets $A\subseteq V$, there exists a countable set $A_0\subseteq A$ such that $\sup_{\succeq}A_0=\sup_{\succeq}A$.

\medskip

\begin{lemma}\label{latt_sup}
$(L^0(\F),\geq_{\mu\textnormal{-a.e.}})$ is a Dedekind complete Riesz space with the countable sup property.
\end{lemma}
\begin{proof}
See \cite[Lemma 2.6.1]{meyer}.
\end{proof}

\smallskip

Interestingly, $L^p$-spaces are ideals of $L^0(\F)$ and thus they are Dedekind complete Riesz spaces with the countable sup property.\footnote{See pp.\ 115-116 of  \cite{meyer}.} In particular, the result we will employ is formalized as follows.

\smallskip

\begin{lemma}\label{L1_count_sup}
$(L^1(\F),\geq_{\mu\textnormal{-a.e.}})$ is a Dedekind complete Riesz space with the countable sup property.
\end{lemma}

\begin{proof}
The countable sup property follows from Lemma \ref{latt_sup}. Indeed, the $\geq_{\mu\textnormal{-a.e.}}$-supremum of any $\geq_{\mu\textnormal{-a.e.}}$-bounded subset $A$ of $L^1(\F)$ is $\F$-measurable. Moreover, proving that $L^1(\F)$ is solid would imply that $L^1(\F)$ is Dedekind complete by Lemma \ref{ideals}, thereby yielding that $\sup_{\geq_{\mu\textnormal{-a.e.}}}A$ has finite $L^1$-norm being $\geq_{\mu\textnormal{-a.e.}}$-bounded. However, the solidity of $L^1(\F)$ follows immediately from the monotonicity of the Lebesgue integral. 
\end{proof}

\smallskip

In the following we will ease the notation for $\geq_{\mu\textnormal{-a.e.}}$, using simply $\geq$ whenever the relevant measure will be clear from the context.

\medskip

\subsection{A Hahn-Banach-Kantorovich Extension Theorem}
Given two real vector spaces $V$ and $W$, we denote by $L(V,W)$ the space of linear mappings $\ell:V\to W$. If $(W,\succeq)$ is a partially ordered vector space, then we say that a mapping $h:V\to W$ is $\succeq$-sublinear if for all $x,y\in V$ and $\alpha\geq 0$, 
$$h(x)+h(y)\succeq h(x+y) \ \ \hbox{and} \ \ h(\alpha x)=\alpha h(x).$$ 

\smallskip

\noindent Given two functions $f$ and $g$ from a real vector space $V$ to a partially ordered vector space $(W,\succeq)$, we say that $f\succeq g$ if $f(x)\succeq g(x)$ for all $x\in V$.

\medskip

\begin{theorem}\label{HBgen}
Let $V$ be a real vector space, $(W,\succeq)$ a Dedekind complete Riesz space, and $h:V\to W$ a $\succeq$-sublinear function. If $S$ is a
vector subspace of $V$ and $\ell:S\to W$ is a linear function with $h|_S\succeq \ell$, then there exists $\ell^*\in L(V,W)$ such that $\ell^*|_S=\ell$ and $h\succeq \ell^*$.
\end{theorem}

\begin{proof}
See \cite[Theorem 1.25]{AliBurki06}.
\end{proof}

\smallskip

We say that a function $f:V\to W$ from a partially ordered vector space $(V,\succeq_V)$ to a partially ordered vector space $(W,\succeq_W)$ is $(\succeq_V,\succeq_W)$-\textit{monotone} if for all $x,y\in V$, 
$$x\succeq_V y \ \implies \  f(x)\succeq_W f(y).$$

\smallskip

\noindent Note that for a linear mapping $\ell:V\to W$, $(\succeq_V,\succeq_W)$-monotonicity is the same as $(\succeq_V,\succeq_W)$-positivity, i.e., $\ell$ is monotone if and only if $\ell(x)\succeq_W\mathbf{0}_W$, for all $x\succeq_V\mathbf{0}_V$.

\medskip

\begin{theorem}\label{HBKgen}
Let $(V,\succeq_V)$ and $(W,\succeq_W)$ be respectively a partially ordered vector space and a Dedekind complete Riesz space. Let $h:V\to W$ be a $(\succeq_V,\succeq_W)$-monotone and $\succeq_W$-sublinear function, and suppose that $S$ is a
vector subspace of $V$. If $\ell:S\to W$ is a $(\succeq_V,\succeq_W)$-positive linear function with $h|_S\succeq_W \ell$, then there exists a $(\succeq_V,\succeq_W)$-positive $\ell^*\in L(V,W)$ such that $\ell^*|_S=\ell$ and $h\succeq_W \ell^*$.
\end{theorem}

\begin{proof}
By Theorem \ref{HBgen}, $\ell$ admits a linear extension $\ell^*$ with $h\succeq_W\ell^*$. This implies that for all $x\succeq_V \mathbf{0}_V$ we have $\mathbf{0}_W=h(\mathbf{0}_V)\succeq_W h(-x)\succeq_W \ell^*(-x)=-\ell^*(x)$, which yields $\ell^*(x)\succeq_W \mathbf{0}_W$. Thus, $\ell^*$ is $(\succeq_V,\succeq_W)$-positive.
\end{proof}

\medskip
\section{Representations of Sublinear Mappings}
\label{AppRep}

\subsection{Envelope Representation Results}

Here we state abstract versions of the envelope representation for sublinear mappings.

\smallskip

\begin{theorem}\label{HBrep}
Let $V$ be a real vector space and $(W,\succeq)$ be a Dedekind complete Riesz space. If $h:V\to W$ is a $\succeq$-sublinear mapping, then
$$h(x)=\sup\limits_{\ell\in D(h)}\ell(x), \ \forall x\in V,$$
where $D(h)=\left\lbrace \ell \in L(V,W):h\succeq \ell \right\rbrace$.
\end{theorem}

\smallskip

\begin{proof}
Fix $x\in V$ and define $l:\textnormal{span}\left\lbrace x\right\rbrace \to W$ by $l(\alpha x):=\alpha h(x)$, for all $\alpha\in \R$. First notice that, $l$ is well-defined. If $x=\mathbf{0}_V$, then the claim follows from the fact that $h(\mathbf{0}_V)=\mathbf{0}_W$. Therefore, suppose that $x\neq \mathbf{0}_V$. Assume that $\alpha x=\beta x$ for some $\alpha,\beta\in \R$. If $\alpha,\beta\geq 0$, by the positive homogeneity of $h$ the claim is immediate. If $\alpha,\beta<0$, then $-\alpha x=-\beta x$, and hence, by positive homogeneity,
$$-\alpha h(x)=h(-\alpha x)=h(-\beta x)=-\beta h(x),$$
which gives $l(\alpha x)=\alpha h(x)=\beta h(x)=l(\beta x)$. Note that if $\alpha\geq 0$ and $\beta<0$, then $(\alpha-\beta)x=\mathbf{0}_V$, and since $\alpha-\beta>0$ we would have $x=\mathbf{0}_V$, a contradiction.

\smallskip

We now show that $h(\alpha x)\succeq l(\alpha x)$ for all $\alpha\in \R$. Since $h$ is $\succeq$-sublinear and for all $\beta\in \R$, $\beta+|\beta|\geq 0$, it follows that 
$$\alpha h(x)+|\alpha|h(x)=h((\alpha+|\alpha|)x)\leq  h(\alpha x)+|\alpha|h(x),$$
for all $\alpha\in \R$. This implies that for all $\alpha\in \R$, we have $l(\alpha x)=\alpha h(x)\leq h(\alpha x)$. Thus, $h|_{\textnormal{span}\left\lbrace x\right\rbrace}\succeq l$. Our goal now is to prove that $l$ is linear. To this end, notice that $l(\alpha (\beta x))=l(\alpha\beta x)=\alpha\beta h(x)=\alpha l(\beta x)$, for all $\alpha,\beta\in \R$, proving that $l$ is homogeneous. Analogously, note that $l(\alpha x+\beta x)=(\alpha+\beta)h(x)=l(\alpha x)+l(\beta x)$. Thus, $l$ is linear. Therefore, by Theorem \ref{HBgen}, it follows that there exists $l^*:V\to W$ such that $l^*|_{\textnormal{span}\left\lbrace x\right\rbrace}=l$, $h\succeq l^*$, and $l^*(x)=l(x)=h(x)$. Since the choice of $x\in V$ is arbitrary, these steps show that for all $z\in V$, there exists a function $\ell_z\in D(h)$ such that $\ell_z(z)=h(z)$. Thus, $\sup_{\ell\in D(h)}l(z)\succeq h(z)$, for all $z\in V$. Since each $\ell\in D(h)$ is such that $h\succeq \ell$, it follows that $h(z)\succeq \sup_{l\in D(h)}\ell(z)$, for all $z\in V$. Given that $\succeq$ is antisymmetric, it follows that
$$h(z)= \sup_{l\in D(h)}\ell(z),$$
for all $z\in V$.
\end{proof}

\medskip

\begin{corollary}\label{mon_abs_rep}
Let $(V,\succeq_V)$ be a partially ordered real vector space and $(W,\succeq_W)$ a Dedekind complete Riesz space. If $h:V\to W$ is a $(\succeq_V,\succeq_W)$-monotone and $\succeq_W$-sublinear mapping, then
$$h(x)=\sup\limits_{\ell\in D_+(h)}\ell(x), \ \forall x\in V,$$

\smallskip

\noindent where $D_+(h):=\left\lbrace \ell \in L(V,W):h\succeq_W \ell\ \textnormal{and}\ \ell\ \textnormal{is}\ (\succeq_V,\succeq_W)\textnormal{-positive} \right\rbrace$.
\end{corollary}

\smallskip

\begin{proof}
The proof is identical to that of Theorem \ref{HBrep}, but with the observation that the function $l$ in that proof is $(\succeq_V,\succeq_W)$-positive. Then applying Theorem \ref{HBKgen} in place of Theorem \ref{HBgen} yields the desired result.
\end{proof}

\smallskip

Given $W\subseteq V$, a mapping $f:V\to W$ is said to be $W$-\textit{normalized} if $f(w)=w$, for all $w\in W$. 

\smallskip

\begin{proposition}\label{partial_add}
Let $(V,\succeq)$ be Dedekind complete Riesz space and $W\subseteq V$ a Dedekind complete Riesz subspace. If $h:V\to W$ is a $\succeq$-monotone, $\succeq$-sublinear, and $W$-normalized mapping, then for all $x\in V$,
$$h(x)=\sup\limits_{\ell\in D^W_+(h)}\ell(x),$$ 

\smallskip

\noindent where $D^W_+(h):=\left\lbrace \ell \in L(V,W):h\succeq \ell,\ \ell\ \textnormal{is}\ \succeq\textnormal{-positive},\ \textnormal{and}\ \ell(w)=w \ \textnormal{for\ all\ }w\in W\right\rbrace$.
\end{proposition}

\smallskip

\begin{proof}
By Corollary \ref{mon_abs_rep}, $h(z)=\sup_{\ell\in D_+(h)}\ell(z),$ for all $z\in V$. Since $\ell$ is linear, $h\succeq\ell$, and $h$ is $\succeq$-sublinear we have that
$$w=-h(-w)\preceq \ell(w)\preceq h(w)=w.$$
Since $\succeq$ is antisymmetric it follows that $\ell(w)=w$, for all $w\in W$. This implies that
$$h(z)=\sup\limits_{\ell\in D^W_+(h)}\ell(z),$$
for all $z\in V$.
\end{proof}

\smallskip

Proposition \ref{partial_add} is a first step towards the representation of conditional sublinear expectations. An important observation is that, in all previous results, the sets $D(h)$, $D_+(h)$, and $D^W_+(h)$ are convex. In addition, it can be observed that these sets are the largest sets of linear functionals whose envelope represents $h$. We show this only for $D^W_+(h)$, since the other proofs are analogous.

\medskip

\begin{lemma}\label{uniqueness_D}
Let $(V,\succeq)$ be Dedekind complete Riesz space and $W\subseteq V$ a Dedekind complete Riesz subspace. If there exists a set $D$ of linear, $\succeq$-monotone, and $W$-normalized mappings $\ell:V\to W$ such that
$$h(x)=\sup\limits_{\ell\in D}\ell(x),$$ 
for all $x\in V$, then $D\subseteq D^W_+(h)$ and 
$$h(x)=\sup\limits_{\ell\in D^W_+(h)}\ell(x),$$ 
for all $x\in V$.
\end{lemma}

\smallskip

\begin{proof}
By hypothesis, $\ell(x)\preceq h(x)$ for all $\ell\in D$ and $x\in V$, and hence $\ell \in D^W_+(h)$. This implies that $D\subseteq D^W_+(h)$ and
$$h(x)=\sup\limits_{\ell\in D}\ell(x)\preceq \sup\limits_{l\in D^W_+(h)}l(x),$$
for all $x\in V$. Since for all $l\in  D^W_+(h)$ we must have $l\preceq h$, it follows that $\sup\limits_{l\in D^W_+(h)}l(x)\preceq h(x)$, for all $x\in V$. Given that $\succeq$ is antisymmetric, it follows that
$$h(x)=\sup\limits_{\ell\in D^W_+(h)}\ell(x),$$ 
for all $x\in V$.
\end{proof}

\smallskip

In all above representation results, it is immediate to see that whenever $h$ is linear, the set of dominated linear mappings is a singleton whose only element is $h$. We report this straightforward observation as a lemma for referential purposes.

\medskip

\begin{lemma}\label{linear_lemma}
Let $(V,\succeq)$ be Dedekind complete Riesz space and $W\subseteq V$ a Dedekind complete Riesz subspace.
If $h\in L(V,W)$, then $D(h)=\left\lbrace h \right\rbrace$.
\end{lemma}

\smallskip

\begin{proof}
Suppose that $\ell$ is a linear mapping dominated by $h$. Then, $-h(-x)\preceq \ell(x)\preceq h(x)$, for all $x\in V$. Since $h$ is linear, it follows that $h=\ell$.
\end{proof}

\medskip

Clearly, the same result holds for $h$ being a linear, $\succeq$-monotone, and $W$-normalized mapping, that is, $D^W_+(h)=\left\lbrace h \right\rbrace$.

\medskip
\subsection{Representation Based on Conditional Expectation}
Consider a probability space $(\Omega,\F,\mathbb{P})$. Random variables in this space are identified with their equivalence classes, and hence all inequalities and equalities are intended to hold $\mathbb{P}$-almost surely. Let $\G$ be a sub-$\sigma$-algebra of $\F$. We will refer to $L^1(\G)$-normalization as $\G$-normalization. We denote $\X=L^1(\F)$. A function $h:\X\to L^1(\G)$ is \textit{continuous from above} if for all sequences $(X_m)_{m\in \mathbb{N}}\in \X^{\mathbb{N}}$ such that $|X_m|\leq Y$, for all $m\in \mathbb{N}$, and for some $Y\in \X$, 
$$X_m\downarrow X \ \implies \ h(X_m)\downarrow h(X).$$
It is immediate to note that whenever $h$ is sublinear, monotone, and continuous from above, all positive linear mappings $\ell$ dominated by $h$ must also be continuous from above. Indeed, for the same sequence as before $(X_m)_{m\in \mathbb{N}}$, we have clearly that $X_m-X\downarrow 0$ and hence,
\[
0\leq \ell(X_m)-\ell(X)=\ell(X_m-X)\leq h(X_m-X)\downarrow h(0)=0
\]
that implies that $\ell(X_m)\downarrow \ell(X)$. Note also that because of linearity, all such $\ell$ will also be continuous from below.

\medskip

\begin{proposition}\label{main_representation}
If $h:\X\to L^1(\G)$ is monotone, sublinear, $\G$-normalized, and continuous from above, then there exists a convex set $\mathcal{Q}$ of probability measures on $(\Omega, \F)$ such that
$$h(X)=\sup\limits_{Q\in \mathcal{Q}}\mathbb{E}^{Q}\left[X|\G\right],$$

\smallskip

\noindent for all $X\in \X$. Moreover, $Q\ll\mathbb{P}$ and  $Q|_{\mathcal{G}}=\mathbb{P}|_{\mathcal{G}}$, for all $Q \in \mathcal{Q}$.
\end{proposition}

\smallskip

\begin{proof}
By Proposition \ref{partial_add} and the previous observation on continuity, we have
$$h(X)=\sup\limits_{\ell\in D^{\G}_+(h)}\ell(X),$$
for all $X\in \X$, where $D^{\G}_+(h)$ is the set of all linear, continuous from below, monotone, and $\G$-normalized mappings $\ell$ such that $\ell\leq h$. Fix $\ell\in  D^{\G}_+(h)$. We now show that there exists a probability measure $Q$ on $(\Omega,\mathcal{F})$ such that
$$\ell:X\mapsto \mathbb{E}^{Q}\left[X|\mathcal{G}\right].$$

\noindent Define the set function $Q$ on $\F$ by $Q(A)=\displaystyle\int_{\Omega}\ell\left(\mathbf{1}_A\right)\mathrm{d}\mathbb{P}$, for all $A\in \mathcal{F}$. Since $\Omega,\emptyset\in \mathcal{G}$ and $\ell$ is $\mathcal{G}$-normalized, $Q(\Omega)=\displaystyle\int_{\Omega}\ell(\mathbf{1}_{\Omega})\mathrm{d}\mathbb{P}=\mathbb{P}(\Omega)=1$. Analogously, $Q(\emptyset)=0$. By linearity of $\ell$, the additivity of $Q$ is immediate. Now, let $\left(A_m\right)_{m\in \mathbb{N}}$ be a disjoint sequence of elements of $\F$. Then by linearity, monotonicity, and continuity from below of $\ell$, the Monotone Convergence Theorem yields
\begin{align*}
\sum_{m=1}^{\infty}Q(A_m)&=\lim\limits_{m\to \infty}Q\left(\bigcup_{k=1}^mA_k\right)=\lim\limits_{m\to \infty}\int_{\Omega}\ell\left(\mathbf{1}_{\bigcup_{k=1}^mA_k}\right)\mathrm{d}\mathbb{P}\\
&=\int_{\Omega}\ell\left(\lim\limits_{m\to \infty}\mathbf{1}_{\bigcup_{k=1}^mA_k}\right)\mathrm{d}\mathbb{P}=\int_{\Omega}\ell\left(\mathbf{1}_{\bigcup_{n=1}^{\infty}A_n}\right)\mathrm{d}\mathbb{P}=Q\left(\bigcup_{n=1}^{\infty}A_n\right).
\end{align*}
Thus $Q$ is a probability measure over $\mathcal{F}$. Moreover, since $\ell$ is $\mathcal{G}$-normalized, it follows that for all $A\in \mathcal{G}$, $Q(A)=\displaystyle\int_{\Omega}\ell(\mathbf{1}_A)\mathrm{d}\mathbb{P}=\int_{\Omega}\mathbf{1}_A\mathrm{d}\mathbb{P}=\mathbb{P}(A)$ and hence, $Q|_{\mathcal{G}}=\mathbb{P}|_{\mathcal{G}}$. By definition of $Q$ it is immediate to see that $Q\ll\mathbb{P}$.

\medskip

Now, for all $A\in \F$, $\ell(\mathbf{1}_A)$ is $\G$-measurable. Therefore, since $Q|_{\mathcal{G}}=\mathbb{P}|_{\mathcal{G}}$, it follows that
$$\int_{\Omega}\mathbf{1}_A\mathrm{d}Q=\int_{\Omega}\ell\left(\mathbf{1}_{A}\right)\mathrm{d}\mathbb{P}=\int_{\Omega}\ell(\mathbf{1}_{A})\mathrm{d}Q.$$
By linearity and continuity from below of $\ell$, the Monotone Convergence Theorem implies that 
\begin{equation}\label{eq:linearrep}
\mathbb{E}^Q\left[X\right]=\mathbb{E}^Q\left[\ell(X)\right]=\mathbb{E}\left[\ell(X)\right],
\end{equation}
for all $X\in \X$. By \eqref{eq:linearrep}, we have that if $X\in \mathcal{X}$, then
\[
\mathbb{E}^Q\left[|X|\right]=\mathbb{E}^Q\left[\ell(|X|)\right]=\mathbb{E}\left[\ell(|X|)\right]<\infty.
\]
This implies that $\mathcal{X}\subseteq  L^1\left(\Omega,\mathcal{F},Q\right)$. Since $\mathcal{X}=L^1\left(\Omega,\mathcal{F},\mathbb{P}\right)$ we have that:
\begin{enumerate}
\item $\mathbb{R}\subseteq \mathcal{X}$,
\medskip
\item If $X,Y\in \mathcal{X}$ and $A_1,A_2\in \F$ with $A_1\cap A_2=\emptyset$, then $\mathbf{1}_{A_1}X+\mathbf{1}_{A_2}Y\in \mathcal{X}$.
\end{enumerate}
Moreover, by definition $\ell(\mathcal{X})\subseteq L^1(\G)$ and since $Q$ is absolutely continuous with respect to $\mathbb{P}$, we have that $\ell(Z)=Z$ $Q$-a.s. for all $Z\in L^1(\G)$. By monotonicity and expectation invariance (i.e., \eqref{eq:linearrep}) of $\ell$ with respect to $Q$, Theorem 1 in \cite{Pfanzagl} implies that
\[
\ell(\mathbf{1}_A X)=\mathbf{1}_A \ell(X)
\]
for all $X\in \mathcal{X}$ and $A\in \G$. This, together with \eqref{eq:linearrep}, readily implies that
\[
\int_A X\mathrm{d}Q=\int_{\Omega}\mathbf{1}_A X\mathrm{d}Q=\int_{\Omega}\ell(\mathbf{1}_AX)\mathrm{d}Q=\int_A \ell(X)\mathrm{d}Q
\]
for all $X\in \mathcal{X}$ and $A\in \G$. Thus, for all $X\in \mathcal{X}$, we have $\ell(X)=\mathbb{E}^{Q}\left[X|\mathcal{G}\right].$ Since $\ell$ was chosen arbitrarily, it follows that there exists a set of probability measures $\mathcal{Q}\subseteq \bigtriangleup^{\mathbb{P}}(\mathcal{F}|\mathcal{G})$ such that
$$h(X)= \sup\limits_{Q\in \mathcal{Q}}\mathbb{E}^{Q}\left[X|\mathcal{G}\right],$$
for all $X\in \X$. It remains to show that $\mathcal{Q}$ can be taken to be convex. To this end, let $\tilde{\mathcal{Q}} :=\left\lbrace Q\in \bigtriangleup(\F):\forall X\in \X,\ h(X)\geq \mathbb{E}^{Q}\left[X|\mathcal{G}\right] \right\rbrace.$ Then $h(X)=\sup_{Q\in \tilde{\mathcal{Q}}}\mathbb{E}^{Q}\left[X|\mathcal{G}\right]$, for all $X\in \X$, and $\tilde{Q}$ is convex. 
\end{proof}

\newpage
\section{Proofs and Auxiliary Results}
\label{AppProofs}

For $\Xvec\in \X^n$ and a permutation $\pi:\left[n\right]\to \left[n\right]$, we denote by $\Xvec_{\pi}$ the vector $\left(X_{\pi(1)},\ldots,X_{\pi(n)}\right)$.

\medskip

\begin{lemma}\label{abs_agentsano}
Suppose that $H$ is a robust allocation mechanism. The following are equivalent:

\medskip

\begin{enumerate}[label=(\roman*)]
\item $H$ satisfies AA.

\medskip

\item $D_i\left(H_i,S,\G\right)=D_j\left(H_j,S,\G\right)$ for all $i,j\in \left[n\right]$, $S\in \X$, and $\G\in \Sigma$.
\end{enumerate}
\end{lemma}

\smallskip

\begin{proof}
\noindent\underline{$\left[(i)\implies (ii)\right]$}: Suppose that $H$ satisfies AA. Let $i,j\in \left[n\right]$ with $i\neq j$, $X,S\in \X$, $\G\in \Sigma$, and $\pi:\left[n\right]\to \left[n\right]$ be a permutation with $\pi(j)=i$. Fix $\Xvec$ and $\Yvec$ in $\X^n$ with $X_i=X$, $X_j=S-X$,  $Y_j=X$, $Y_i=S-X$, and zero in all other entries. Clearly $\Yvec=\Xvec_{\pi}$ and $\G^{\Xvec}=\G^{\Yvec}$. By AA and Lemma \ref{uniqueness_D}
\begin{equation}\label{inf_equations}
\inf\limits_{L\in D_i(H_i,S,\G)}L(X)=H_{\pi(j)}\left(\Xvec|\G\right)=H_{j}\left(\Xvec_{\pi}|\G\right)=H_{j}\left(\Yvec|\G\right)=\inf\limits_{L\in D_j(H_j,S,\G)}L(X).
\end{equation}

\medskip

\noindent Since the choice of $X\in \X$ was arbitrary, it follows that 
\begin{align*}
D_i(H_i,S,\G)&=\left\lbrace \tilde{L}\in \mathcal{L}:\forall X\in \X, \tilde{L}(X)\geq \inf\limits_{L\in D_i(H_i,S,\G)}L(X) \right\rbrace\\
&=\left\lbrace \tilde{L}\in \mathcal{L}:\forall X\in \X, \tilde{L}(X)\geq \inf\limits_{L\in D_j(H_j,S,\G)}L(X) \right\rbrace=D_j(H_j,S,\G).
\end{align*}

\bigskip

\noindent\underline{$\left[(ii)\implies (i)\right]$}: Suppose that $D_i(H_i,S,\G)=D_j(H_j,S,\G)$, for all $i,j\in \left[n\right]$, $S\in \X$, and $\G\in \Sigma$. Let $\Xvec\in \X^n$, $\G\in \Sigma$, and $\pi:\left[n\right]\to \left[n\right]$ be a permutation. Then
$$H_{\pi(i)}(\Xvec|\G)=\inf\limits_{L\in D_{\pi(i)}(H_{\pi(i)},S^{\Xvec},\G)}L(X_{\pi(i)})=\inf\limits_{L\in D_i(H_{i},S^{\Xvec},\G)}L(X_{\pi(i)})=H_i(\Xvec_{\pi}|\G).$$

\medskip

\noindent Since the choice of $i\in \left[n\right]$ and $\Xvec\in \X^n$ was arbitrary, AA holds. 
\end{proof}

\medskip

\begin{corollary}\label{agentsano}
Suppose that $H$ is a robust mean allocation mechanism, and $\bigtriangleup\subseteq \bigtriangleup(\F)$ is  nonempty. If for all $i\in \left[n\right]$,
\[
\mathcal{Q}_i(S,\G)=\left\lbrace Q\in \bigtriangleup:\forall X\in \X,\ \mathbb{E}^{Q}\left[X|\G^S\right]\geq \inf\limits_{Q_i\in \mathcal{Q}_i\left(S,\G\right)}\mathbb{E}^{Q_i}\left[X|\mathcal{G}^S\right] \right\rbrace,
\]
then, the following are equivalent:

\medskip

\begin{enumerate}[label=(\roman*)]
\item $H$ satisfies Agent Anonymity.

\bigskip

\item $\mathcal{Q}_i(S,\G)=\mathcal{Q}_j(S,\G)$, for all $i,j\in \left[n\right]$, $S\in \X$, and $\G\in \Sigma$.
\end{enumerate}
\end{corollary}

\smallskip

\begin{proof}
This follows from Lemma \ref{abs_agentsano}. 
\end{proof}

\medskip

\begin{lemma}\label{cost_sub_is_split_robu}
An allocation mechanism $H$ satisfies FP if and only if for all $\Yvec\in \X^n$, $i,j\in \left[n\right]$, and $\G\in \Sigma$,

\vspace{-0.15cm}

\begin{equation}\label{sp_proof}
H_i(\Xvec|\G)+H_j(\Xvec|\G)\geq H_i(\Yvec|\G)+H_j(\Yvec|\G),
\end{equation}

\medskip

\noindent where $\Xvec:=\Yvec-Y_j\evec^{(j)}+Y_j\evec^{(i)}$.
\end{lemma}

\smallskip

\begin{proof} 
Suppose that $H$ satisfies FP. Fix $\Yvec\in \X^n$, $i\neq j$ in $\left[n\right]$, and let $\Xvec:=\Yvec-Y_j\evec^{(j)}+Y_j\evec^{(i)}$. Clearly, $X_j=0$ and hence, by Frictional Participation,

\vspace{-0.3cm}

\[
Y_i+Y_j-H_i(\Yvec|\G)-H_j(\Yvec|\G)\geq X_i+X_j-H_i(\Xvec|\G)-H_j(\Xvec|\G).
\]

\medskip

\noindent Therefore, since $X_j=0$ and $X_i=Y_i+Y_j$, it follows that condition \eqref{sp_proof} holds.

\medskip

Conversely, suppose that $H$ satisfies \eqref{sp_proof}. Fix $\Xvec\in \X^n$ with $X_j=0$ for some $j\in \left[n\right]$, and fix $Z\in \X$. Let $\Yvec:=\Xvec+Z\evec^{(j)}-Z\evec^{(i)}$. Then $\Xvec=\Yvec-Y_j\evec^{(j)}+Y_j\evec^{(i)}$. Since $Y_i+Y_j=X_i-Z+Z=X_i+X_j$, condition \eqref{sp_proof} yields

\vspace{-0.15cm}

\[
Y_i+Y_j-H_i(\Yvec|\G)-H_j(\Yvec|\G)\geq X_i+X_j-H_i(\Xvec|\G)-H_j(\Xvec|\G),
\]

\medskip

\noindent and hence $H$ satisfies FP.
\end{proof}

\smallskip
\subsection{Proof of Theorem \ref{RGRS_prop}} \ 

\smallskip

\noindent\underline{$\left[(i)\implies (ii)\right]$}:  Let $S\in \X$ and $\G\in \Sigma$. Define $h^{S|\G}:\X\to \X$ by
\[h^{S|\G}(X):=H_1\left(X,S-X,0,\ldots,0|\G\right),\] 
for all $X\in \X$. We now verify the following two claims.

\medskip

\begin{enumerate}
\item \underline{$h^{S|\G}$ is superadditive}. To prove superadditivity fix $X,Y\in \X$. Then
\begin{align*}
h^{S|\G}\left(X+Y\right)&=H_1\left(X+Y,S-X-Y,0,\ldots,0|\G\right)\\
&\geq H_1\left(X,S-X-Y,Y,0\ldots,0|\G\right)+H_3\left(X,S-X-Y,Y,0\ldots,0|\G\right)\\
&=H_1\left(X,S-X-Y,Y,0\ldots,0|\G\right)+H_1\left(Y,S-X-Y,X,0\ldots,0|\G\right)\\
&=H_1\left(X,S-X,0\ldots,0|\G\right)+H_1\left(Y,S-Y,0\ldots,0|\G\right)
=h^{S|\G}(X)+h^{S|\G}(Y),
\end{align*}

\medskip

\noindent where the inequality follows from FP, Lemma \ref{cost_sub_is_split_robu}, and ZP. The second equality follows from AA, and the second-to-last equality follows from OA.

\medskip

\item \underline{$h^{S|\G}$ is positively homogeneous}. By SI, it follows that for all $\alpha\in \left[0,1\right]$ and $X\in \X$, we have 
\begin{align*}
h^{S|\G}(\alpha X)
&=H_1\left(\alpha X,S-\alpha X,0,\ldots,0|\G\right)
=H_1\left(\alpha X,S- X,(1-\alpha)X,\ldots,0|\G\right)\\
&=H_3\left((1-\alpha)X,S- X,\alpha X,\ldots,0|\G\right)
=\alpha H_1\left(X,S-X,0,\ldots,0|\G\right)=\alpha h^{S|\G}(X),
\end{align*}


\noindent where the second and third equalities follow from OA and AA, respectively. Thus, $h^{S|\G}$ is positively homogeneous.
\end{enumerate}

\smallskip

\noindent Consequently, and since the choice of $S\in \X$ was arbitrary, Theorem \ref{HBrep} implies that
$$H_1(\Xvec|\G)=H_1\left(X_1,S^{\Xvec}-X_1,0,\ldots,0|\G\right)=h^{S^{\Xvec}|\G}(X_1)=\inf\limits_{L\in D_1\left(H_1,S^{\Xvec},\G\right)}L(X_1),$$
where $$D_1\left(H_1,S^{\Xvec},\G\right)=\bigg\{ L\in \mathcal{L}:\forall X\in \X,\ L(X)\geq H_1(X,S^{\Xvec}-X,0,\ldots,0|\G) \bigg\}.$$ 

\medskip

\noindent Proceeding similarly, we obtain that
$$H_j(\Xvec|\G)=\inf\limits_{L\in D_j\left(H_j,S^{\Xvec},\G\right)}L(X_j), \ \forall j \in \{2,\ldots,n\}.$$

\medskip

\noindent\underline{$\left[(ii)\implies (i)\right]$}: For all $i\in [n]$, we denote 
\[
D\left(H_i,S^{\Xvec},\G\right)=D_i\left(H_i,S^{\Xvec},\G\right),
\]

\medskip

\noindent for all $i\in \left[n\right]$, $\Xvec\in \X^n$, and $\G\in \Sigma$.

\medskip

\begin{enumerate}
    \item \underline{FP}. Let $\Xvec\in \X$, $i,j\in \left[n\right]$, and $\Yvec=\Xvec-X_j\evec^{(j)}+X_j\evec^{(i)}$. Then
    \begin{align*}
H_i(\Yvec|\G)&=\inf\limits_{L\in D\left(H_i,S^{\Xvec},\G\right)}L(X_i+X_j)\geq \inf\limits_{L\in D\left(H_i,S^{\Xvec},\G\right)}L(X_i)+\inf\limits_{L\in D\left(H_i,S^{\Xvec},\G\right)}L(X_j)\\
&= \inf\limits_{L\in D\left(H_i,S^{\Xvec},\G\right)}L(X_i)+\inf\limits_{L\in D\left(H_j,S^{\Xvec},\G\right)}L(X_j)=H_i(\Xvec|\G)+H_j(\Xvec|\G).
\end{align*}
By Lemma \ref{cost_sub_is_split_robu}, $H$ satisfies FP.
\medskip  
    \item \underline{SI}. For all $\G \in \Sigma$, all $\Xvec \in \X^n$, and all $i,j \in [n]$, if $X_j=0$, $Z =\alpha X_i$ for some $\alpha \in [0,1]$, and $\Yvec := \Xvec + Z \evec^{(j)} - Z \evec^{(i)}$, then
    \begin{align*}
\alpha H_i(X_1, \ldots, X_i, \ldots, X_{j-1}, 0, X_{j+1}, \ldots, X_n|\G)&=\inf\limits_{L\in D_i\left(H_i,S,\G\right)}L(\alpha X_i)\\
&=\inf\limits_{L\in D_j\left(H_j,S,\G\right)}L(\alpha X_i)=H_j(\Yvec|\G).
    \end{align*}

\medskip
    
    \item \underline{OA} follows from the observation that for all allocations $\Xvec\in \X^n$, each $H_i$ depends only on $X_i$, $S^{\Xvec}$, and $\G$.

\medskip
    
    \item \underline{AA} follows from Lemma \ref{abs_agentsano}.

\medskip

    \item \underline{ZP} follows from the positive homogeneity of each $H_i$. \qed
\end{enumerate}

\medskip
\subsection{Proof of Proposition \ref{PropConvexCost}} \ 

\smallskip

\noindent Follows immediately from the proof of Theorem \ref{RGRS_prop}. \qed

\medskip
\subsection{Proof of Theorem \ref{thm1}} \ 

\smallskip

\noindent\underline{$\left[(i)\implies (ii)\right]$}: Let $S\in \X$ and $\G\in \Sigma$. Define $h^{S|\G}:\X\to \X$ by 
$$h^{S|\G}(X):=H_1\left(X,S-X,0,\ldots,0|\G\right),$$
for all $X\in \X$. Following the same steps as in the proof of Theorem \ref{RGRS_prop}, one can show that $h^{S|\G}$ is superadditive and positively homogeneous. We now verify the following claims.

\medskip

\begin{enumerate}
\item \underline{$h^{S|\G}$ is monotone}. Suppose that $X\geq Y$ in $\X$. By IF, we have that
$$H_1(X,Y,S-X-Y,0,\ldots,0)\geq H_2(X,Y,S-X-Y,0,\ldots,0).$$
By OA and AA, we have
\begin{align*}
h^{S|\G}(X)&=H_1(X,Y,S-X-Y,0,\ldots,0)\geq H_2(X,Y,S-X-Y,0,\ldots,0)\\
&=H_1(Y,X,S-X-Y,0,\ldots,0)=h^{S|\G}(Y),
\end{align*}
proving that $h^{S|\G}$ is monotone.

\medskip

\item The fact that \underline{$h^{S|\G}(X)$ is $\G^S$-measurable}, for all $X\in \X$, follows from IA.

\medskip

\item \underline{$h^{S|\G}$ is $\G^S$-normalized}. If $X$ is $\G^S$-measurable, then $(X,S-X,0,\ldots,0)$ is also $\G^S$-measurable and hence, by IB, we have
$$h^{S|\G}(X)=X.$$

\smallskip

\item \underline{$h^{S|\G}$ is continuous from above}. It follows directly from CA.
\end{enumerate}

\medskip

\noindent From Proposition \ref{main_representation}, it follows that there exists a set of probability measures $\mathcal{Q}_1(S,\G)$ such that
$$h^{S|\G}(X)= \inf\limits_{Q\in \mathcal{Q}_1(S,\G)}\mathbb{E}^{Q}\left[X|\mathcal{G}^S\right],$$
for all $X\in \X$. Then by AA, OA, and Corollary \ref{agentsano}, the proof of sufficiency is concluded.

\medskip

\noindent\underline{$\left[(ii)\implies (i)\right]$}: Suppose that $H$ is a robust conditional mean allocation mechanism with $\mathcal{Q}_i(S,\G)=\mathcal{Q}_j(S,\G)$, for all $S\in \X$, $i,j\in \left[n\right]$, and $\G\in \Sigma$. FP, SI, and OA were already proved for Theorem \ref{RGRS_prop}.

\medskip

\begin{enumerate}
    \item \underline{AA} follows from Corollary \ref{agentsano}.

\medskip
    
    \item \underline{IF}. Let $\Xvec\in \X^n$ and $\G\in \Sigma$. If $X_i\geq X_j$, then
    \begin{align*}
    H_i(\Xvec|\G)&= \inf\limits_{Q\in \mathcal{Q}_i(S,\G)}\mathbb{E}^{Q}\left[X_i|\mathcal{G}^{\Xvec}\right]=\inf\limits_{Q\in \mathcal{Q}_j(S,\G)}\mathbb{E}^{Q}\left[X_i|\mathcal{G}^{\Xvec}\right]\\
    &\geq \inf\limits_{Q\in \mathcal{Q}_j(S,\G)}\mathbb{E}^{Q}\left[X_j|\mathcal{G}^{\Xvec}\right]= H_j(\Xvec|\G).
    \end{align*}

\medskip
    
    \item \underline{IB}. Let $\Xvec\in \X^n$ and $\G\in \Sigma$. If $\Xvec$ is $\G^{\Xvec}$-measurable, then each $X_i$ is $\G^{\Xvec}$-measurable. This implies by the property of conditional expectations that
    $$H_i(\Xvec|\G)= \inf\limits_{Q\in \mathcal{Q}_i(S,\G)}\mathbb{E}^{Q}\left[X_i|\mathcal{G}^{\Xvec}\right]= \inf\limits_{Q\in \mathcal{Q}_i(S,\G)}X_i=X_i,$$
    for all $i\in \left[n\right]$. 

\medskip
    
    \item \underline{IA}. Let $\Xvec\in \X^n$, $\G\in \Sigma$, and $i\in \left[n\right]$. By definition, $H_i(\Xvec|\G)\in \X$. Thus, $\inf\limits_{Q\in \mathcal{Q}_i(S^{\Xvec},\G)}\mathbb{E}^{Q}\left[X_i|\mathcal{G}^{\Xvec}\right]\in \X$. This implies that the set
    \[
    \bigg\{ \mathbb{E}^{Q}\left[X_i|\mathcal{G}^{\Xvec}\right]:Q\in \mathcal{Q}_i(S^{\Xvec},\G) \bigg\}
    \]
    is order bounded with respect to the $\mathbb{P}$-a.s. dominance order. Therefore, by Lemma \ref{L1_count_sup}, there exists a countable set $\hat{\mathcal{Q}}\subseteq \mathcal{Q}_i(S^{\Xvec},\G)$ such that 
    \begin{equation}\label{count_sup}
    \inf\limits_{Q\in \mathcal{Q}_i(S^{\Xvec},\G)}\mathbb{E}^{Q}\left[X_i|\mathcal{G}^{\Xvec}\right]=\inf\limits_{Q\in \hat{\mathcal{Q}}}\mathbb{E}^{Q}\left[X_i|\mathcal{G}^{\Xvec}\right].
    \end{equation}

\medskip

    \noindent Since each $ \mathbb{E}^{Q}\left[X_i|\mathcal{G}^{\Xvec}\right]$ is $\G^{\Xvec}$-measurable, it follows from the measurability of infima of countable sets  and equation \eqref{count_sup} that
\[
H_i(\Xvec|\G)=\inf\limits_{Q\in \hat{\mathcal{Q}}}\mathbb{E}^{Q}\left[X_i|\mathcal{G}^{\Xvec}\right]\in \mathcal{X}\left(\G^{\Xvec}\right).
\]
By the arbitrary choice of $\Xvec\in \X^n$, $\G\in \Sigma$, and $i\in \left[n\right]$, IA follows.

\medskip
    
    \item \underline{CA}. Let $S\in \X$, $(X_m)_{m\in \mathbb{N}}\in \X^{\mathbb{N}}$ with $X_m\xrightarrow{\mathbb{P}-a.s.}X$ for some $X\in \X$ and $\lvert X_m\rvert\leq Y$ for some $Y\in \X$ and all $m\in \mathbb{N}$. Fix $\G\in \Sigma$. Since $\mathcal{Q}_i(S,\G)\subseteq \bigtriangleup^{\mathbb{P}}(\F|\G)$, we have that $X_m\xrightarrow{Q-a.s.}X$ for all $Q\in \mathcal{Q}_i(S,\G)$. Then, by Dominated Convergence Theorem, it follows that for all $i\neq j$
    {\small\begin{align*}
    H_i(X\evec^{(i)}+(S-X)\evec^{(j)}|\G)&=\inf\limits_{Q\in \mathcal{Q}_i(S,\G)}\mathbb{E}^{Q}\left[X|\mathcal{G}^{S}\right]=\inf\limits_{Q\in \mathcal{Q}_i(S,\G)}\mathbb{E}^{Q}\left[\limsup\limits_{m\to\infty}X_m|\mathcal{G}^{S}\right]\\
    &=\inf\limits_{Q\in \mathcal{Q}_i(S,\G)}\limsup\limits_{m\to\infty}\mathbb{E}^{Q}\left[X_m|\mathcal{G}^{S}\right]=\inf\limits_{Q\in \mathcal{Q}_i(S,\G)}\inf\limits_{m\geq 0}\sup\limits_{k\geq m}\mathbb{E}^{Q}\left[X_m|\mathcal{G}^{S}\right]\\
    &=\inf\limits_{m\geq 0}\inf\limits_{Q\in \mathcal{Q}_i(S,\G)}\sup\limits_{k\geq m}\mathbb{E}^{Q}\left[X_m|\mathcal{G}^{S}\right]\geq \inf\limits_{m\geq 0}\sup\limits_{k\geq m}\inf\limits_{Q\in \mathcal{Q}_i(S,\G)}\mathbb{E}^{Q}\left[X_m|\mathcal{G}^{S}\right]\\
    &=\limsup\limits_{m\to \infty}\inf\limits_{Q\in \mathcal{Q}_i(S,\G)}\mathbb{E}^{Q}\left[X_m|\mathcal{G}^{S}\right]=\limsup\limits_{m\to \infty}H_i(X_m\evec^{(i)}+(S-X_m)\evec^{(j)}|\G). 
    \end{align*}}
    Given the arbitrariety of $S\in \X$, $(X_m)_{m\in \mathbb{N}}\in \X^{\mathbb{N}}$, $\G\in \Sigma$, $i,j\in \left[n\right]$, we have that all the functions
    \[
    h^S_{i,j}:Z\mapsto H_i(Z\evec^{(i)}+(S-Z)\evec^{(j)}|\G)
    \]
    satisfy $h^S_{i,j}(Y)\geq \limsup\limits_{m\to\infty} h^S_{i,j}(Y_m)$, for all $i,j\in \left[n\right]$, $\G\in \Sigma$,  $(Y_m)_{m\in \mathbb{N}}\in \X^{\mathbb{N}}$ with $Y_m\xrightarrow{\mathbb{P}-a.s.}Y$ with $|Y_m|\leq Z$ for some $Y,Z\in \X$ and $S\in \X$. Now we prove that each $h^S_{i,j}$ is continuous from above. To this end, suppose that $(Z_m)_{m\in \mathbb{N}}\in \X^{\mathbb{N}}$ and $Z\in \X$, $Z_m\downarrow Z$, and $|Z_m|\leq Y$, for all $m\in \mathbb{N}$, for some $Y\in \X$. By IF, each $h^S_{i,j}$ is a monotone function. Hence, by the previous steps, it follows that
$$\limsup\limits_{m\to\infty}h_{i,j}^S(Z_m)\leq h_{i,j}^S(Z)\leq \liminf\limits_{m\to\infty}h_{i,j}^S(Z_m).$$

\medskip

\noindent Thus, $h_{i,j}^S(Z_m)\downarrow h_{i,j}^S(Z)$.\qed
\end{enumerate}

\medskip
\subsection{Proof of Proposition \ref{subj_CMRS}} \ 

\smallskip

\noindent\underline{$\left[(i)\implies (ii)\right]$}: The first steps of the proof follow the same reasoning as the proof of Theorem \ref{RGRS_prop}. In particular, let $S\in \X$ and $\G\in \Sigma$. Define $h^{S|\G}:\X\to \X$ by 
$$h^{S|\G}(X):=H_1\left(X,S-X,0,\ldots,0|\G\right),$$
for all $X\in \X$. The main difference in the sufficiency part of the proof is that FP* leads to additivity rather than superadditivity. This implies that $h^{S|\G}$ is linear. Repeating the arguments in the proof of Theorem \ref{RGRS_prop}, we obtain that $h^{S|\G}$ is also monotone, $\G^S$-normalized, positively homogeneous, and $h^{S|\G}(X)$ is $\G^S$-measurable for all $X\in \X$. By order continuity from below and linearity it follows that $h^{S|\G}$ is continuous from above. By Proposition \ref{main_representation} and Lemma \ref{linear_lemma}, there exists $Q\in \bigtriangleup^{\mathbb{P}}(\F|\G)$ such that
$$h^{S|\G}(X)=\mathbb{E}^{Q}\left[X|\G\right],$$
for all $X\in \X$. Then by AA, OA, and Corollary \ref{agentsano}, the proof of sufficiency is concluded.

\medskip

\noindent\underline{$\left[(ii)\implies (i)\right]$}: Following the steps of the proof of necessity of Theorem \ref{thm1}, it is immediate to see that any linear conditional mean allocation mechanism satisfies FP*, IF, IB, IA, SI, and AA. CA follows from the Monotone Convergence Theorem for conditional expectation. \qed

\medskip
\subsection{Proof of Proposition \ref{prop_properties}} \ 

\smallskip

\noindent Since $H(0|\cdot)=0$ and $\sigma(S^{\Xvec})=\sigma(\alpha S^{\Xvec})$, for all $\alpha>0$ and all $\Xvec\in \X^n$, it is straightforward to see that $H$ satisfies Positive Homogeneity. Normalization and Constancy follow from observing that constants are measurable with respect to all $\sigma$-algebras, and from the properties of conditional expectations. Since $\sigma(S^{\Xvec})=\sigma(S^{\Xvec}+c)$, for all $c\in \mathbb{R}$ and all $\Xvec\in \X^n$, it is immediate to see that $H$ satisfies Translativity. Suppose now that $\Xvec\in \X^n$. If $\sup X_j=+\infty$, then the inequality for the No-Ripoff property trivially holds. Suppose that $\sup X_j\in \mathbb{R}$.  The monotonicity of conditional expectations and the Constancy property proved above imply that $H$ satisfies the No-Ripoff property.\qed

\medskip

\subsection{A Comparison with Other Risk-Sharing Rules}

To compare with other existing rules, assume that $\mathcal{G}=\left\lbrace \emptyset,\Omega\right\rbrace$. First we define two risk-sharing rules that will be compared with our robust conditional mean-risk sharing rule. In particular, we define the Conditional Mean Risk-Sharing (CMRS) rule and Quantile-Based Risk-Sharing (QBRS) rule. The CMRS rule is simply defined as follows:
\[
H^{\textnormal{CMRS}}:\Xvec\mapsto \mathbb{E}\left[\Xvec\mid S^{\Xvec}\right].
\]

\medskip

In order to define the QBRS rule is necessary to introduce some notation. For all $Y\in \mathcal{X}$ denote by $F_{Y}$ the cumulative distribution function (CDF) of $Y$, and for all $\alpha,p\in [0,1]$ define:
\medskip
\begin{align*}
&F^{-1}_Y(p)=\inf\left\lbrace y\in \R:F_Y(y)\geq p \right\rbrace\\
&F^{-1,+}_Y(p)=\sup\left\lbrace y\in \R:F_Y(y)\leq p \right\rbrace\\
&F^{-1,\alpha}_Y(p)=\begin{cases}
F^{-1,+}_Y(0) & \textnormal{if}\ p=0\\
\alpha F^{-1}_Y(p)+(1-\alpha)F^{-1,+}_Y(p) & \textnormal{if}\ p\in (0,1)\\
F^{-1}_Y(p) & \textnormal{if}\ p=1
\end{cases}
\end{align*}

\medskip

\noindent Given a vector of initial endowments $\Xvec=(X_1,\ldots,X_n)\in \mathcal{X}^n$ define its comonotonic counterpart $\Xvec^c:=\left(F_{X_1}^{-1}(U),\ldots,F_{X_n}^{-1}(U)\right)$, where $U$ is a random variable uniformly distributed on $[0,1]$. The QBRS rule is defined by:
\[
H^{\textnormal{QBRS}}:\Xvec\mapsto \left(F^{-\!1\left(\alpha_{S^{\Xvec}}\right)}_{X_i}\!\Big(F_{S^{\Xvec,c}}\!\left(S^{\Xvec}\right)\Big)\right)_{i=1}^n,
\]
where $\alpha_{S^{\Xvec}}\in[0,1]$ is determined implicitly by
\[
F^{-\!1\left(\alpha_{S^{\Xvec}}\right)}_{\,S^{\Xvec,c}}\!\Big(F_{S^{\Xvec,c}}\!\left(S^{\Xvec}\right)\Big)
\;=\;
S^{\Xvec}.
\]

\medskip

Table \ref{tab:axiom_matrix} provides a comparison of the CMRS, QBRS, and RCMRS.

\medskip

\begin{table}[H]
\centering
\caption{Comparison of risk sharing rules}
\label{tab:axiom_matrix}
\renewcommand{\arraystretch}{1.15}
\setlength{\tabcolsep}{6pt}
\begin{tabular}{lcccccccccc}
\toprule
\textbf{Rule / Property} 
& Com & UI & IF & AF & RF & ZP & AA & OA & IA \\
\midrule
CMRS        & \ding{55} & \ding{51} & \ding{51} & \ding{51} & \ding{51} & \ding{51} & \ding{51} & \ding{51} & \ding{51}  \\
QBRS        & \ding{51} & \ding{55} & \ding{51} & \ding{55} & \ding{51} & \ding{51} & \ding{51} & \ding{55} & \ding{51}  \\
RCMRS       & \ding{55} & \ding{55} & \ding{51} & \ding{55} & \ding{51} & \ding{51} & \ding{51} & \ding{51} & \ding{51}  \\
\bottomrule
\end{tabular}
\end{table}

\noindent In Table \ref{tab:axiom_matrix}, 
\begin{itemize}
\item Com: denotes \textit{comonotonicity}. A risk-sharing rule $H$ satisfies comonotonicity if, for all $\Xvec \in \mathcal{X}^n$, the allocation $H(\Xvec)$ is comonotone. This property is particularly relevant in risk sharing due to its close link with Pareto optimality under risk aversion. Although Pareto efficiency cannot be assessed without specifying the agents' preferences, examining whether a rule yields comonotonic allocations provides an indirect means of evaluating Pareto optimality of the induced allocation for certain classes of preferences. The QBRS rule satisfies Com, while CMRS and RCRMS rules in general do not satisfy Com.

\medskip

\item UI: denotes \textit{universal improvement}. A risk-sharing rule $H$ satisfies universal improvement if, for all $\Xvec \in \mathcal{X}^n$ and $i \in [n]$, it holds that $X_i \succeq_{\textnormal{cx}} H_i(\Xvec)$.\footnote{Recall that $\succeq_{\textnormal{cx}}$ denotes the convex order: $X \succeq_{\textnormal{cx}} Y$ iff $\mathbb{E}[\varphi(X)] \geq \mathbb{E}[\varphi(Y)]$ for all convex functions $\varphi : \mathbb{R} \to \mathbb{R}$.} As with Pareto optimality, participation incentives of risk sharing rules cannot generally be evaluated without specifying the agents’ preferences. Nevertheless, the UI property provides a useful benchmark: if a rule satisfies UI, then risk-averse expected-utility agents prefer to participate. The CMRS rule satisfies UI, while QBRS and RCRMS rules in general do not satisfy UI.

\medskip

\item Fairness conditions: IF, AF, ZP are fairness conditions that we already discussed in Section \ref{SecRep}. RF denotes \textit{risk fairness}, sometimes also termed \textit{no-ripoff}, and it requires that the allocation to each agent should not exceed their maximum
possible loss. AF (actuarial fairness) is only satisfied by CMRS.

\medskip

\item Anonymity conditions: AA, OA, IA are anonymity conditions that we already discussed in Section \ref{SecRep}. The CMRS, RCNRS, and QBRS all satisfy properties AA and IA. The QBRS is the only risk-sharing rule that does not satisfy OA, as there is no risk-sharing rule on $\X^n$ that satisfies OA, Com, and ZP (Proposition 5 in Jiao et al.\ \cite{Jiaoetal2023}).
\end{itemize}

\medskip

Formally, the literature does not provide a standard definition of a \textit{robust} risk-sharing rule. Nonetheless, model uncertainty implies that relying on a single probabilistic specification $\mathbb{P}$ may yield misspecified assessments. In this perspective, CMRS and QBRS are not robust, whereas RCMRS rules offer a robust alternative.

\medskip

\subsubsection{Table \ref{tab:axiom_matrix} Formal Explanation}\ 

\begin{itemize}
\item CMRS: By Theorem 1, Corollary 1, and Proposition 6 in Jiao et al.\ \cite{Jiaoetal2023} the CMRS rule satisfies OA, AA, AF, RF, UI. It is straightforward to see that it also satisfies IF and ZP. As a consequence, it does not satisfy Com (Proposition 5 in Jiao et al.\ \cite{Jiaoetal2023}).

\medskip

\item QBRS: By Proposition 5.6 in Denuit et al.\ \cite{Denuitetal2022}, the QBRS rule satisfies Com, RF, AA, and ZP, but it satisfies  neither AF nor UI. By Proposition 5 in Jiao et al.\ \cite{Jiaoetal2023} ,it does not satisfy OA. By its quantile formulation, it also satisfies IF.

\medskip

\item RCMRS: By Theorem \ref{thm1} and Proposition  \ref{prop_properties}, any RCMRS rule satisfies OA, AA, RF, IF, ZP, IA, and RF. However, it may not satisfy AF (see the explanation in Section \ref{SecRiskSharing}), and consequently it may fail to satisfy UI. Moreover, since CMRS does not satisfy Com, an RCMRS rule may not satisfy Com either. 
\end{itemize}

\label{app:ex52}

\medskip

\subsection{Example \ref{example:meandeviation} Expanded} \ 

\smallskip

\noindent \begin{itemize}
\item $H^{D,\theta}\in \mathcal{AM}$:
\begin{align*}
\sum_{i=1}^nH_{i}^{D,\theta}(\Xvec|\G)
&=S^{\Xvec}- \sum_{i=1}^n\theta D(X_i|\G^{\Xvec})
\leq S^{\Xvec}- \sum_{i=1}^n\theta D\left(\sum_{i=1}^nX_i\middle\vert\G^{\Xvec}\right)
\leq S^{\Xvec},
\end{align*}

\medskip

\noindent where the last inequality follows from $ D(X_i|\G^{\Xvec})\geq 0$ and $\theta\geq 0$.

\medskip

\item IF: Since $\mathbb{E}\left[\cdot|\G\right]$ and $-D(\cdot|\G)$ are increasing, it follows that $H^{D,\theta}$ satisfies Internal Fairness.

\item AA: For all permutations $\pi$,
\[
H_{\pi(i)}^{D,\theta}(\Xvec|\G)=\mathbb{E}\left[X_{\pi(i)}|\G^{\Xvec}\right]-\theta D(X_{\pi(i)}|\G^{\Xvec})=H_{i}^{D,\theta}(\Xvec_{\pi}|\G).
\]


\item OA: This follows from observing that each $H_{i}^{D,\theta}$ depends on all $X_j$ for $j\neq i$ only through $S^{\Xvec}$.

\medskip

\item FP: This follows from the superadditivity of $\mathbb{E}\left[\cdot|\G\right]-\theta D(\cdot|\G)$, for all $\G\in \Sigma$.

\medskip

\item SI: This follows from the positive homogeneity $\mathbb{E}\left[\cdot|\G\right]-\theta D(\cdot|\G)$, for all $\G\in \Sigma$.

\medskip

\item ZP: This follows from the positive homogeneity $\mathbb{E}\left[\cdot|\G\right]-\theta D(\cdot|\G)$, for all $\G\in \Sigma$.

\medskip

\item IB: This follows from the fact that $D(X|\G)=0$, for all $\G\in \Sigma$ and $\G$-measurable $X\in \X$.

\medskip

\item IA: This follows from the $\G$-measurability of each $\mathbb{E}\left[X|\G\right]-\theta D(X|\G)$. \qed
\end{itemize}

\medskip

\subsection{Example \ref{example:ES} Expanded}\ 

\smallskip
\label{app:ex52}
\noindent 
We start recalling that, for all $i\in [n]$,
\[
\sigma_{1:n}^2=\sum_{i=1}^n\sigma^2_i+2\sum_{i<j}\rho_{ij}\sigma_i\sigma_j\ \textnormal{and}\ \bar{\rho}_i=\frac{\sum_{j=1}^n\rho_{ij}\sigma_i\sigma_j}{\sigma_i\sigma_{1:n}}.
\]
The computations of the examples are based on the following observations:
\[
X_1+X_2+\cdots+X_n\sim \mathcal{N}\left(\mu_1+\cdots+\mu_n,\sigma^2_{1:n}\right).
\]

\smallskip

\noindent Moreover, recall that whenever $X\sim \mathcal{N}(\mu_X,\sigma^2_X)$ and $Y\sim \mathcal{N}(\mu_Y,\sigma^2_Y)$, it follows that
\[
X|Y=y\sim \mathcal{N}\left(\mu_X+\frac{\sigma_X}{\sigma_Y}\rho_{XY}(y-\mu_Y),(1-\rho^2_{XY})\sigma^2_X\right),
\]
which implies that
\begin{equation}\label{eq:conditionalnormaldist}
X_i|X_1+\cdots+X_n=x\sim \mathcal{N}\left(\mu_i+\frac{\sigma_i}{\sigma_{1:n}}\bar{\rho}_i(x-\mu_1+\cdots+\mu_n),(1-\bar{\rho}^2_i)\sigma^2_i\right).
\end{equation}

\medskip

\noindent  Now, recall that for $X\sim \mathcal{N}(\mu,\sigma)$ and all $\lambda\in (0,1)$, 
\[
\inf\left\lbrace x\in \R:\mathbb{P}\left[X\leq x\right]\geq \lambda\right\rbrace = \mu+\sigma \, \Phi^{-1}(\lambda),
\]

\smallskip

\noindent
where $\Phi$ is the standard normal CDF. Consequently, denoting by $\varphi$ the density of the standard normal, we obtain
\begin{align*}
\mu+\frac{\sigma}{\lambda}\int_0^{\lambda}\Phi^{-1}(\gamma)\textnormal{d}\gamma &= \mu+\frac{\sigma}{\lambda}\int_{\Phi^{-1}(0)}^{\Phi^{-1}(\lambda)}\Phi^{-1}(\Phi(\upsilon))\varphi(\upsilon)\textnormal{d}\upsilon=\mu+\frac{\sigma}{\lambda}\int_{\Phi^{-1}(0)}^{\Phi^{-1}(\lambda)}\upsilon\varphi(\upsilon)\textnormal{d}\upsilon\\
&=\mu+\frac{\sigma}{\lambda}\int_{-\infty}^{\Phi^{-1}(\lambda)}\upsilon\varphi(\upsilon)\textnormal{d}\upsilon = \mu-\frac{\sigma}{\lambda\sqrt{2\pi}}\left[\textnormal{exp}\left\lbrace -y^2/2 \right\rbrace\right]^{\Phi^{-1}(\lambda)}_{-\infty}\\
&=\mu-\frac{\sigma}{\lambda}\frac{\textnormal{exp}\left\lbrace -\Phi^{-1}(\lambda)^2/2 \right\rbrace}{\sqrt{2\pi}}=\mu-\sigma\frac{\varphi(\Phi^{-1}(\lambda))}{\lambda}.
\end{align*}

\noindent
Therefore, by \eqref{eq:conditionalnormaldist}, it follows that
\begin{align*}
H_i^{\lambda}(\Xvec|\G)
&=\frac{1}{\lambda}\int_0^{\lambda}\inf\left\lbrace x\in \R:\mathbb{P}\left[X\leq x|S^{\Xvec}\right]\geq \gamma\right\rbrace\textnormal{d}\gamma\\
&=\mu_{X_i|S^{\Xvec}}-\frac{\sigma_{X_i|S^{\Xvec}}}{\lambda}\int_0^{\lambda}\Phi^{-1}(\gamma)\textnormal{d}\gamma
=\mu_{X_i|S^{\Xvec}}-\sigma_{X_i|S^{\Xvec}}\frac{\varphi(\Phi^{-1}(\lambda))}{\lambda}\\
&=\mu_i+\frac{\sigma_i\bar{\rho}_i}{\sigma_{1:n}}(S^{\Xvec}-\mu_1+\cdots+\mu_n)-\frac{\varphi(\Phi^{-1}(\lambda))}{\lambda}\sqrt{(1-\bar{\rho}^2_i)\sigma^2_i}.
\end{align*}

\noindent Consequently,
\begin{align*}
C^{\Xvec,\G}(H)&=S^{\Xvec}-\sum_{i=1}^nH_i^{\lambda}(\Xvec|\G)\\
&=S^{\Xvec}-\sum_{i=1}^n\mu_i-\sum_{i=1}^n\frac{\sigma_i\bar{\rho}_i}{\sigma_{1:n}}\left(S^{\Xvec}-\sum_{i=1}^n\mu_i\right)+\sum_{i=1}^m\frac{\sqrt{(1-\bar{\rho}^2_i)\sigma^2_i}}{\lambda}\varphi(\Phi^{-1}(\lambda))\\
&=\left(S^{\Xvec}-\sum_{i=1}^n\mu_i\right)\left(1-\sum_{i=1}^n\frac{\sigma_i\bar{\rho}_i}{\sigma_{1:n}}\right)+\sum_{i=1}^n\frac{\sqrt{(1-\bar{\rho}^2_i)\sigma^2_i}}{\lambda}\varphi(\Phi^{-1}(\lambda))\\
&=\frac{\varphi(\Phi^{-1}(\lambda))}{\lambda}\sum_{i=1}^n\sqrt{(1-\bar{\rho}^2_i)\sigma^2_i},
\end{align*}

\smallskip

\noindent where the last equality follows from
\[
\sum_{i=1}^n\sigma_i \bar{\rho}_i=\sum_{i=1}^n\sigma_i\left(\sum_{j=1}^n\frac{\rho_{ij}\sigma_i\sigma_j}{\sigma_i\sigma_{1:n}}\right)=\sum_{i=1}^n\sum_{j=1}^n\frac{\rho_{ij}\sigma_i\sigma_j}{\sigma_{1:n}}=\sigma_{1:n}.
\]

\medskip

\noindent We define the individual friction cost of agent $i$ as $C^{\Xvec,\G}_{i}(H)=X_i-H_i(X|\G)$, for all $i\in [n]$. In the case considered, this cost becomes
\[
C_i^{\Xvec,\G}(H)=X_i-\mu_i-\frac{\sigma_i\bar{\rho}_i}{\sigma_{1:n}}(S^{\Xvec}-\mu_1+\cdots+\mu_n)+\frac{\varphi(\Phi^{-1}(\lambda))}{\lambda}\sqrt{(1-\bar{\rho}^2_i)\sigma^2_i}.
\]

\medskip

\noindent
We denote the expected individual frictional costs by $\bar{C}_i^{\Xvec,\G}(H):=\mathbb{E}[C_i^{\Xvec,\G}(H)]$. In the example considered above, we obtain
\[
\bar{C}_i^{\Xvec,\G}(H)=\frac{\varphi(\Phi^{-1}(\lambda))}{\lambda}\sqrt{(1-\bar{\rho}^2_i)\sigma^2_i}.
\]

\medskip

\noindent
Assume now that individuals have mean-variance preferences given by $V_i:X\mapsto \mathbb{E}[X]-\theta_i \sigma^2_X$. Note that
\begin{align*}
V_i(H_i(X|\G))-V_i(X)&=-\bar{C}_i^{\Xvec,\G}(H)-\theta_i\left[\sigma^2_{H_i^\lambda(X|\G)}-\sigma^2_{i}\right]\\
&=-\frac{\varphi(\Phi^{-1}(\lambda))}{\lambda}\sqrt{(1-\bar{\rho}^2_i)\sigma^2_i}-\theta_i\left[\sigma^2_{H_i^\lambda(X|\G)}-\sigma^2_{i}\right]\\
&=-\frac{\varphi(\Phi^{-1}(\lambda))}{\lambda}\sqrt{(1-\bar{\rho}^2_i)\sigma^2_i}-\theta_i\left[\sigma^2_i\bar{\rho}^2_i-\sigma^2_{i}\right]\\
&=\sqrt{(1-\bar{\rho}^2_i)\sigma^2_i}\left[\theta_i\sqrt{(1-\bar{\rho}^2_i)\sigma^2_i}-\frac{\varphi(\Phi^{-1}(\lambda))}{\lambda}\right].
\end{align*}

\bibliographystyle{plain}
\bibliography{biblio}

@article{ghossoub2026efficiency,
  title={{Efficiency in Pure-Exchange Economies With Risk-Averse Monetary Utilities}},
  author = "M. {\textsc{Ghossoub}} and M.B. {\textsc{Zhu}}",
  journal={Mathematical Finance},
  volume={36},
  number={1},
  pages={99--117},
  year={2026},
  publisher={Wiley Online Library}
}

@book{mcneil2015quantitative,
  title={Quantitative risk management: concepts, techniques and tools-revised edition},
  author={A. {\textsc{McNeil}}  and R. {\textsc{Frey}} and R{\"u}diger and P. {\textsc{Embrechts}}},
  year={2015},
  publisher={Princeton University Press}
}

@article{embrechts2018quantile,
  title={Quantile-based Risk Sharing},
  author= "P. {\textsc{Embrechts}} and H. {\textsc{Liu}} and R. {\textsc{Wang}}",
  journal={Operations Research},
  volume={66},
  number={4},
  pages={936--949},
  year={2018},
  publisher={INFORMS}
}

@article{embrechts2020quantile,
  title={Quantile-based Risk Sharing with Heterogeneous Beliefs},
  author="P. {\textsc{Embrechts}} and H. {\textsc{Liu}} and T. {\textsc{Mao}} and R. {\textsc{Wang}}",
  journal={Mathematical Programming},
  volume={181},
  pages={319--347},
  year={2020},
  publisher={Springer}
}

@article{filipovic2008optimal,
  title={{Optimal Capital and Risk Allocations for Law-and Cash-Invariant Convex Functions}},
  author="D. {\textsc{Filipovi{\'c}}} and G. {\textsc{Svindland}}",
  journal={Finance and Stochastics},
  volume={12},
  pages={423--439},
  year={2008},
  publisher={Springer}
}

@article{JouiniSchachermayerTouzi2008,
author = "E.\ {\textsc{Jouini}} and W.\ {\textsc{Schachermayer}} and N.\ {\textsc{Touzi}}",
journal = "Mathematical Finance",
title = {{Optimal Risk Sharing for Law Invariant Monetary Utility Functions}},
pages = "269–292",
volume = "18",
number="2",
year = "2008"
}

@article{DanaLeVan2010,
author = "R.A. {\textsc{Dana}} and C. {\textsc{Le Van}}",
title = {{Overlapping Risk Adjusted Sets of Priors and the Existence of Efficient Allocations and Equilibria with Short-Selling}},
journal = "Journal of Economic Theory",
year = "2010",
volume = "145",
number = "6",
pages = "2186-2202"
}

@article{RavanelliSvindland2014,
author = "C. {\textsc{Ravanelli}} and G. {\textsc{Svindland}}",
title = {{Pareto Optimal Allocations for Law Invariant Robust Utilities on $L^1$}},
journal = "Finance and Stochastics",
year = "2014",
volume = "18",
pages = "249–269"
}

@article{Malinvaud1972,
author = "E. {\textsc{Malinvaud}}",
title = {{The Allocation of Individual Risks in Large Markets}},
journal = "Journal of Economic Theory",
year = "1972",
volume = "4",
number = "2",
pages = "312-328"
}

@article{Malinvaud1973,
author = "E. {\textsc{Malinvaud}}",
title = {{Markets for an Exchange Economy with Individual Risks}},
journal = "Econometrica",
year = "1973",
volume = "41",
number = "3",
pages = "383-410"
}

@article {RockafellarUryasev06,
author = "K.T. {\textsc{Rockafellar}} and S. {\textsc{Uryasev}} and M. {\textsc{Zabarankin}}",
     TITLE = {{Generalized Deviations in Risk Analysis}},
   JOURNAL = {Finance and Stochastics},
  FJOURNAL = {Finance and Stochastics},
    VOLUME = {10},
      YEAR = {2006},
    NUMBER = {1},
     PAGES = {51--74},
      ISSN = {0949-2984,1432-1122},
   MRCLASS = {90C25 (60A10 90C46 90C47 90C56 91B28 91B30)},
  MRNUMBER = {2212567},
MRREVIEWER = {Salvador\ C.\ Rambaud},
       DOI = {10.1007/s00780-005-0165-8},
       URL = {https://doi.org/10.1007/s00780-005-0165-8},
}

@article{BBG24,  
 title = "({N}o-){B}etting {P}areto-{O}ptima under {R}ank-{D}ependent {U}tility",
author = "P. {\textsc{Bei{\ss}ner}} and T.J. {\textsc{Boonen}} and M. {\textsc{Ghossoub}}",
journal = "Mathematics of Operations Research",
year = "2024",
volume = "49",
number = "3",
pages = "1452-1471"
}

@article{BeissnerWerner2023,
author = "P. {\textsc{Beissner}} and J. {\textsc{Werner}}",
title = {{Optimal Allocations with $\alpha$-MaxMin Utilities,
Choquet Expected Utilities, and Prospect Theory}},
journal = "Theoretical Economics",
year = "2023",
volume = "18",
number = "3",
pages = "993–1022"
}

@article{Pfanzagl,
 ISSN = {00034851},
 URL = {http://www.jstor.org/stable/2239153},
 author = {J. {\textsc{Pfanzagl}}},
 journal = {The Annals of Mathematical Statistics},
 number = {2},
 pages = {415--421},
 publisher = {Institute of Mathematical Statistics},
 title = {{Characterizations of Conditional Expectations}},
 urldate = {2023-11-19},
 volume = {38},
 year = {1967}
}

@unpublished{BionNadal04,
author = "J. {\textsc{Bion-Nadal}} and F. {\textsc{Palaiseau}}",
title = {{Conditional Risk Measure and Robust Representation of Convex Conditional Risk Measures}},
note = "Mimeo",
year = "2004"
}

@article {FilKupVogel2012,
author = "D. {\textsc{Filipovi\'{c}}} and M. {\textsc{Kupper}} and N. {\textsc{Vogelpoth}}",
title = {{Approaches to Conditional Risk}},
journal = "SIAM Journal on Financial Mathematics",
volume = "3",
year = {2012},
number = "1",
pages = "402--432",
}

@article {DetlefsenScandolo2005,
author = "K. {\textsc{Detlefsen}} and G. {\textsc{Scandolo}}",
title = {{Conditional and Dynamic Convex Risk Measures}},
journal = "Finance and Stochastics",
year = "2005",  
volume = "9",
number = "4",
pages = "539-561",
}

@article {Bartl2020,
author = "D. {\textsc{Bartl}}",
title = {{Conditional Nonlinear Expectations}},
journal = "Stochastic Process and their Applications",
year = "2020",    
volume = "130",
number = "2",
pages = "785-805",
}

@article{Dana2002,
author = "R.A. {\textsc{Dana}}",
journal = "Annals of Operations Research",
title = "{O}n {E}quilibria when {A}gents have {M}ultiple {P}riors",
year = "2002",
volume = "114",
number = "1",
pages = "105-115"
}

@article{Rigottietal2008,
author = "L. {\textsc{Rigotti}} and C. {\textsc{Shannon}} and T. {\textsc{Strzalecki}}",
title = "{S}ubjective {B}eliefs and ex ante {T}rade",
journal = "Econometrica",
year = "2008",
volume = "76",
number = "5",
pages = "1167-1190"
}

@article{StrzaleckiWerner2011,
author = "T. {\textsc{Strzalecki}} and J. {\textsc{Werner}}",
journal = "Journal of Economic Theory",
title = "{E}fficient {A}llocations under {A}mbiguity",
year = "2011",
volume = "146",
number = "3",
pages = "1173-1194"
}

@article{decastro2011,
  title={{Ambiguity Aversion and Trade}},
  author={L.I. {\textsc{De Castro}} and A. {\textsc{Chateauneuf}}},
  journal={Economic Theory},
  volume={48},
  number = "2/3",
  pages={243-273},
  year={2011}
}

@book{AliBurki06,
    author = "C.D. {\textsc{Aliprantis}} and O. {\textsc{Burkinshaw}}",
     TITLE = {{Positive Operators}},
      NOTE = {{Reprint of the 1985 original}},
 PUBLISHER = "Springer, Dordrecht",
      YEAR = "2006"
}

@unpublished{Dhaeneetal2023,
author = "J. {\textsc{Dhaene}}
and C.J {\textsc{Robert}} and K.C. {\textsc{Cheung}} and M. {\textsc{Denuit}} ",
title = "{A}n {A}xiomatic {T}heory for {Q}uantile-{B}ased
{R}isk {S}haring",
note = "Mimeo",
year = "2023"
}

@unpublished{Jiaoetal2023,
author = "Z. {\textsc{Jiao}} and S. {\textsc{Kou}} and Y. {\textsc{Liu}} and R. {\textsc{Wang}}",
title = "{A}n {A}xiomatic {T}heory for {A}nonymized {R}isk {S}haring",
note = "arXiv: 2208.07533",
year = "2023"
}

@article{Denuitetal2022,
author = "M. {\textsc{Denuit}} and J. {\textsc{Dhaene}} and C.J {\textsc{Robert}}",
title = "{R}isk-{S}haring {R}ules and their {P}roperties, with {A}pplications to {P}eer-to-{P}eer {I}nsurance",
journal = "Journal of Risk and Insurance",
year = "2022",
volume = "89",
number = "3",
pages = "615-667"
}

@article{CarlierDanaGalichon2012,
author = "G. {\textsc{Carlier}} and R.A. {\textsc{Dana}} and A. {\textsc{Galichon}}",
title = "{P}areto {E}fficiency for the {C}oncave {O}rder and {M}ultivariate {C}omonotonicity",
journal = "Journal of Economic Theory",
year = "2012",
volume = "147",
number = "1",
pages = "207-229"
}

@article{DenuitDhaene2012,
author = "M. {\textsc{Denuit}} and J. {\textsc{Dhaene}}",
title = "{C}onvex {O}rder and {C}omonotonic {C}onditional {M}ean {R}isk {S}haring",
journal = "Insurance Mathematics and Economics",
year = "2012",
volume = "51",
number = "2",
pages = "265-270"
}

@article{KajiiUi2006,
author = "A. {\textsc{Kajii}} and T. {\textsc{Ui}}",
title = "{A}greeable {B}ets with {M}ultiple {P}riors",
journal = "Journal of Economic Theory",
year = "2006",
volume = "128",
number = "1",
pages = "299-305"
}

@book{DelbaenOsaka,
    author = "F. {\textsc{Delbaen}}",
    publisher = "Osaka University Lecture Notes Series",
    title = {{Monetary Utility Functions}},
    year = "2012"
}

@article{Wilson1968,
author = "R. {\textsc{Wilson}}",
title = "{T}he {T}heory of {S}yndicates",
journal = "Econometrica",
year = "1968",
volume = "36",
number = "1",
pages = "119-132"
}

@unpublished{CCK97,
author = "A. {\textsc{Chateauneuf}} and M. {\textsc{Cohen}} and R. {\textsc{Kast}}",
title = {{A Review of some Results Related to Comonotonicity}},
year = "1997",
note = "Cahiers Eco \& Maths 97.32, Universit\'e Paris I"
}

@article{Billotetal2000,
author = "A. {\textsc{Billot}} and A. {\textsc{Chateauneuf}} and I. {\textsc{Gilboa}} and J.M. {\textsc{Tallon}}",
title = "{S}haring {B}eliefs: {B}etween {A}greeing and {D}isagreeing",
journal = "Econometrica",
year = "2000",
volume = "68",
number = "3",
pages = "685-694"
}

@article{Billotetal2002,
author = "A. {\textsc{Billot}} and A. {\textsc{Chateauneuf}} and I. {\textsc{Gilboa}} and J.M. {\textsc{Tallon}}",
title = "{S}haring {B}eliefs and the {A}bsence of {B}etting in the {C}hoquet {E}xpected {U}tility {M}odel",
journal = "Statistical Papers",
year = "2002",
volume = "43",
number = "1",
pages = "127-136"
}

@article{Dominiak2012,
  title="{A}greeable {T}rade with {O}ptimism and {P}essimism",
  author="A. {\textsc{Dominiak}} and J. {\textsc{Eichberger}} and J.-P. {\textsc{Lefort}}",
  journal="Mathematical Social Sciences",
  volume="64",
  number="2",
  pages="119-126",
  year="2012"
}

@unpublished{DanaMeilijson2003,
author = "R.A. {\textsc{Dana}} and I. {\textsc{Meilijson}}",
title = "{M}odelling {A}gents' {P}references in {C}omplete {M}arkets by {S}econd {O}rder {S}tochastic {D}ominance",
note = "mimeo (2003)"
}

@article{LandsbergerMeilijson1994c,
author = "M. {\textsc{Landsberger}} and I. {\textsc{Meilijson}}",
journal = "Annals of Operations Research",
title = "{C}o-monotone {A}llocations, {B}ickel-{L}ehmann {D}ispersion and the {A}rrow-{P}ratt {M}easure of {R}isk {A}version",
number = "2",
pages = "97-106",
volume = "52",
year = "1994"
}

@book{Ruschendorf2013,
    author = "L. {\textsc{R\"uschendorf}}",
    publisher = "Springer-Verlag Berlin Heidelberg",
    title = "{M}athematical {R}isk {A}nalysis: {D}ependence, {R}isk {B}ounds, {O}ptimal {A}llocations and {P}ortfolios",
    year = "2013"
}

@article{GhirardatoSiniscalchi2018,
  author="P. {\textsc{Ghirardato}} and M. {\textsc{Siniscalchi}}",
    title="{R}isk {S}haring in the {S}mall and in the {L}arge",
  journal="Journal of Economic Theory",
  volume="175",
  pages="730-765",
  year="2018",
}

@article{Maccheronietal2006,
author = "F. {\textsc{Maccheroni}} and M. {\textsc{Marinacci}} and A. {\textsc{Rustichini}}",
journal = "Econometrica",
title = "{A}mbiguity {A}version, {R}obustnsess, and the {V}ariational {R}epresentation of {P}references",
number = "6",
pages = "1447-1498",
volume = "74",
year = "2006"
}

@book{AliprantisBorder,
    author = "C.D. {\textsc{Aliprantis}} and K.C. {\textsc{Border}}",
    publisher = "Springer-Verlag",
    title = "{I}nfinite {D}imensional {A}nalysis - $3^{rd}$ edition",
    year = "2006"
}

@book{meyer,
    author = "P. A. {\textsc{Meyer}}",
    publisher = "Blaisdell Publishing Company",
    title = "{P}robability and {P}otentials",
    year = "1966"
}

@article{schmeidler89,
    author = "D. {\textsc{Schmeidler}}",
    journal = "Econometrica",
    title = "{S}ubjective {P}robability and {E}xpected {U}tility without {A}dditivity",
    year = "1989",
    volume = "57",
    number = "3",
    pages = "571-587"
}

@article{GilboaSchmeidler1989maxmin,
 author="I. {\textsc{Gilboa}}  and D. {\textsc{Schmeidler}}",
 journal="Journal of Mathematical Economics",
 title="{M}axmin {E}xpected {U}tility with a {N}on-{U}nique {P}rior",
 year="1989",
 volume="18",
 number="2",
 pages="141-153"
}

@article{Chateauneufetal2000,
author = "A. {\textsc{Chateauneuf}} and R.A. {\textsc{Dana}} and J.M. {\textsc{Tallon}}",
journal = "Journal of Mathematical Economics",
title = "{O}ptimal {R}isk-sharing {R}ules and {E}quilibria with {C}hoquet-expected-utility",
number = "2",
pages = "191-214",
volume = "34",
year = "2000"
}

@article{Dana2004,
author = "R.A. {\textsc{Dana}}",
journal = "Economic Theory",
title = "{A}mbiguity, {U}ncertainty {A}version and {E}quilibrium {W}elfare",
number = "3",
pages = "569-587",
volume = "23",
year = "2004"
}

@article{Borch1962a,
author="K. {\textsc{Borch}}",
 journal="Econometrica",
 title="{E}quilibrium in a {R}einsurance {M}arket",
 year="1962",
 volume="30",
 number="3",
 pages="424-44"
}

\vspace{0.2cm}

\end{document}